\documentclass[10pt]{article}

\usepackage{graphicx}
\usepackage{amsmath}
\usepackage{amsfonts}
\usepackage{amsthm}

\newcommand{\dd}{\partial}
\newcommand{\R}{\mathbb{R}}

\renewcommand{\Re}{{\mathcal R e}}
\newcommand{\be}{\begin{equation}}
\newcommand{\ee}{\end{equation}}

\newtheorem{theorem}{Theorem}[section]
\newtheorem{corollary}[theorem]{Corollary}
\newtheorem{proposition}[theorem]{Proposition}
\newtheorem{lemma}[theorem]{Lemma}

\newtheorem{rmk}{Remark}[section]

\DeclareMathOperator{\atan}{atan}

\begin{document}

\begin{center}
  {\Large Critical points of  Strichartz functional} 
\end{center}

\medskip

\centerline{\scshape C. Eugene Wayne
}
\medskip
{\footnotesize
 \centerline{Department of Mathematics, Boston University}
 \centerline{Boston, MA 02215, USA}
}

\medskip

\centerline{\scshape Vadim Zharnitsky
}
\medskip
{\footnotesize
 \centerline{Department of Mathematics, University of Illinois at Urbana-Champaign} 
 \centerline{1409 W. Green Street, Urbana, Illinois 61801-2975, USA}
}

\bigskip

\begin{abstract}

We study a pair of infinite dimensional dynamical systems naturally associated with the study of minimizing/maximizing functions for the Strichartz inequalities for the Schr\"odinger equation. One system is of gradient type and the other one is a Hamiltonian system. For both systems, the corresponding sets of critical points, their stability, and the relation between the two are investigated. By a combination of numerical and analytical methods we argue that the Gaussian is a maximizer in a class of Strichartz inequalities for dimensions one, two and three. The argument reduces to verification of an apparently new combinatorial inequality involving binomial coefficients.

\end{abstract}

\section{Introduction}

Recently there has been considerable interest in the existence and properties of maximizers/minimizers for the Strichartz inequalities.  These are functions which give the best possible constant in these equalities.  One line of research began with Kunze \cite{kunze} who proved the existence of such a function for the one-dimensional Strichartz inequality for solutions of the Schr\"odinger equation.  Later Foschi \cite{foschi} found the value of the best constants in one and two dimensions as well as determining that the maximizing function was given by a Gaussian in both cases.  Foschi's proof was then simplified by Hundertmark and Zharnitsky  \cite{hdvz} who related the maximizing property to orthogonal projections for the space-time norm used to define the Strichartz inequality.

In this paper we propose an alternative approach to study such maximizers based on properties of gradient flows.  We show that the maximizing function is a critical point for a gradient flow in $L^2(\R^n)$.  Because gradient flows are well studied, and because all their orbits must approach a critical point we hope that dynamical systems methods can be used to better understand the properties of these maximizing functions.  To the best of our knowledge, this point of view has not yet been exploited in the search for best constants for various inequalities, and in principle, it should be of use, not just for Strichartz inequalities, on which we focus in this paper, but for other families of inequalities as well.  One line of work that does seem somewhat in the same vein as our own are the papers of Carlen, Carrillo and Loss \cite{carlen} and Bonforte, et al \cite{bonforte}, who relate optimal constants in Hardy-Littlewood-Sobolev and Hardy-Poincar\'e inequalities to solutions of fast diffusion equations.

In addition to the gradient flow we introduce in Section \ref{sec:grad_flow}, the Strichartz inequalities are also naturally related to an infinite dimensional Hamiltonian system.  This has recently been derived in a different context by Faou et al \cite{faou} who showed it arises as a large box limit of a resonant normal form for the NLS equation.  In other recent work, Albert and Kahalil \cite{AK2017} have studied the well-posedness of the Strichartz Hamiltonian flow in one dimension and constructed an example of ill-posedness.  Our work is also related to recent studies of extremizers in the context of Fourier restriction inequalities, see e.g. \cite{Carneiro, ChristShao, ChristQuilodran,FoschiSilva}  and references therein.

In Section \ref{sec:grad_flow}, we will explore the relationship between the Strichartz gradient flow and the Strichartz Hamiltonian flow and investigate in particular, how the latter can shed light on the stability of critical points for the gradient flow.

As an example, we first consider  critical points of the quantum mechanical harmonic oscillator (QMHO). In that system, everything can be explicitly
 calculated and it  will be interesting to 
compare the stability of critical points in the Strichartz functional with those of the QMHO.

  \section{The Hessian for the Quantum-mechanical Harmonic Oscillator Hamiltonian}\label{s:QM}

\medskip

In this section we consider a very simple, explicitly computable example to illustrate our
approach of relating gradient flows and best constants in inequalities.

Consider the quantum-mechanical harmonic oscillator eigenvalue problem:
\begin{equation}
-f_m'' + x^2 f_m = \lambda_m f_m = (2m+1) f_m\ .
\end{equation}

The variational principle for eigenvalues implies that 
\begin{equation}
\frac{H[f]}{\| f \|_{L^2}} = \frac{  \int (x^2 |f|^2 + |f_x|^2) dx}{\| f\|_{L^2}} \ge \lambda_{min} 
\end{equation}
which we can rewrite as
\begin{equation} \label{eq:inequality}
\| f \|_{L^2} \le \frac{1}{\lambda_{min}} H[f]\ .
\end{equation}
From our knowledge of the eigenvalues of the harmonic oscillator, the
``best value'' of the constant on the right hand side of this inequality is ``$1$'',  and
the function that saturates the inequality is the Gaussian.  We now 
illustrate how we could obtain that result from a point of view similar to that we 
will use in the the rest of the paper to study the Strichartz inequality.

%{\color{red}Comment(CEW): The changes in the next couple of paragraphs were made because this
%section started out farther back in the paper and some things (like Hermite functions) which are used
%here hadn't yet been defined.  }

We begin by defining a functional 
\begin{equation} \label{e:functional}
Q[f] = \frac{H[f]}{\int f^2 dx}\ .
\end{equation}
related to this inequality, and study the  gradient flow associated with $Q$.

To actually study this flow it is convenient to expand with respect to the  Hermite functions, $\{ f_n(x)\}$ which form
a basis for $L^2(\R)$.  We will write $f_n(x) = c_n H_n(x) e^{-x^2/2}$, where $H_n$ are the Hermite polynomials
and the normalization constants $c_n$ are chosen so that
\begin{equation}
\int f_n(x) f_m(x) dx = \delta_{n,m}\ .
\end{equation}

If we expand an $L^2$ function as
$$
f = \sum_{m=0} \alpha_m f_m\ ,
$$
then we obtain an expression for $Q$ in terms of $\alpha$ of the form:
\begin{equation}
Q[\alpha] = \frac{ \sum_{m=0} \lambda_m \alpha_m^2}{\sum_{m} \alpha_m^2}\ .
\end{equation}

\begin{rmk} For convenience, in this section we will consider only real valued functions, so
we can assume that the coefficient $\alpha_n$ are real numbers.  It would be straightforward to
extend the following discussion to complex coefficients.
\end{rmk}

Consider the associated gradient flow with
\begin{eqnarray} \label{eq:gradient} \nonumber
\dot{\alpha}_k &=& - \partial_{\alpha_k} Q[\alpha] = \frac{-2 \lambda_k \alpha_k}{\sum_{m} \alpha_m^2}
- \frac{2 \alpha_k  \sum_{m=0} \lambda_m \alpha_m^2}{(\sum_{m} \alpha_m^2)^2}\\ \nonumber
&=& \frac{-2 \lambda_k \alpha_k - 2 \alpha_k Q[\alpha]}{\sum_{m} \alpha_m^2} \\
&=& \frac{-2(\lambda_k - Q[\alpha])}{(\sum_{m} \alpha_m^2)} \alpha_k\ .
\end{eqnarray}

From this formula we can immediately make a number of observations:
\begin{enumerate}
\item For any $n$, the sequence $\alpha_k = \delta_{k,n}$ is a critical point of this flow - i.e. all the 
Hermite functions are critical points.
\item In fact, in this case, we can prove that
these are the only critical points.  Suppose there was a critical point which was not
equal to a Hermite function.  Then its expansion in the Hermite basis would have at least two nonzero
$\alpha_k$'s - say $\alpha_{n_1}$ and $\alpha_{n_2}$.  But then, since $\dot{\alpha}_{n_1}
= \dot{\alpha}_{n_2} =0$ (since we are at a critical point) and hence
$$
Q[\alpha] = \lambda_{n_1}\ , {\mathrm{and}} \ Q[\alpha] = \lambda_{n_2}\ ,
$$
a contradiction, since $n_1 \ne n_2$.  
\item  We can give even more detailed information about the gradient
flow in this instance.   Note that for any non-negative integer $n*$,  the finite
dimensional subspace of $L^2$:
\begin{equation}
S_{n*} = \{ \alpha ~|~ \alpha_k = 0\ , k > n*\}\ ,
\end{equation}
is invariant for the equations of motion \eqref{eq:gradient}.
\item \label{i:unstable} Given initial date $\alpha^0$ for \eqref{eq:gradient}, define $N(\alpha^0) =
\sup\{ n ~|~ \lambda_n \le Q[\alpha^0] \}$.  Then from the equations of motion we see that 
for any $k \le N(\alpha^0)$, $\alpha_k$ is an increasing function of time (or at least, non-decreasing) while
for any $k > n(\alpha^0)$, $\alpha_k$ is a decreasing function of time.  Thus the omega-limit
set for this trajectory lies in the invariant subspace $S_{N(\alpha^0)}$.  Furthermore, since
this is a gradient flow, (and in this case, the very simple form of the equations of motion
allow one to prove that the trajectories lie in compact sets)
the omega-limit set must be a fixed point, so the only possibilities for
the omega-limit set are the critical points $\{ f_0, f_1, \dots f_{N(\alpha^0)}\}$.
Thus, we see that the gradient flow associated to the functional associated to our original
inequality \eqref{eq:inequality} almost always tends toward the function that
yields the best constant in the inequality.  Only if the initial condition happens to
lie in the (finite dimensional) stable manifold of one of the other critical points of the flow
will we fail to reach the optimizing function.
\end{enumerate}

%{\color{red}Comment(CEW):  Added some additional discussion here to explain why we
%work on the submanifold of functions of fixed norm.}

We can also use this functional framework to examine the stability of the critical
points located above.  It is simpler to examine the stability on submanifolds of 
functions of norm one, and we will prove in our discussion of the analogous
computation for the Strichartz inequality below that this is equivalent
to considering the unrestricted variations,  aside from zero eigenvalues
associated with simple invariances of functional.  

Since the Hermite
functions, $f_m$,  corresponding to our critical points are normalized, we 
have
\begin{equation}
H[f_m] = (2m+1)\ .
\end{equation}
Furthermore, on the submanifold of functions of norm one, the denominator of
our functional is always equal to one and we can just look at variations in the numerator.

We now evaluate the Hessian at $f_m$ by inserting the trial function
\begin{equation}
f= \sqrt{1-s_1^2-s_2^2} f_m + s_1 f_k + s_2 f_{\ell}\ .
\end{equation}
Note that this trial function is constructed to insure that it has norm one.

First consider the off-diagonal elements. We find:
\begin{eqnarray}
\frac{\partial^2 H}{\partial s_1 \partial s_2} |_{s_1=s_2=0} &=& 2 \int \left( s^2 f_k f_{\ell} + f_k' f_{\ell}' \right) dx  \nonumber \\
&=& 2 \int \left( x^2 f_k - f_k'' \right) f_{\ell} dx \\
&=& -2 \lambda_k \int f_k f_{\ell} dx = 0\ , \nonumber
\end{eqnarray}
by orthonormality.

Now consider the diagonal terms:
\begin{eqnarray}
\frac{\partial^2 H}{\partial s_1^2} |_{s_1=s_2=0} &=& 2 \left\{ \int \left( (f_k')^2 + x^2 f_k^2 \right) dx
-  \int \left( (f_m')^2 + x^2 f_m^2 \right) dx \right\}  \nonumber \\
&=& 2 \left( (2k+1) - (2m+1) \right) \\
&=& 4(k-m)\ . \nonumber
\end{eqnarray}

Thus, in particular, if we consider the Hessian at the Gaussian, $h_0$, we have all eigenvalues positive, which means that $h_0$ is at least a local minimum, and is consistent with the fact that we know the
Gaussian corresponds to the function giving the smallest possible value of the function $Q[f]$.

The Hessian at the first Hermite function has a single negative eigenvalues meaning that the 
gradient flow has a one dimensional unstable manifold and all other directions are stable.  In addition, 
the discussion in point \ref{i:unstable} above, implies that solutions in the unstable manifold of 
$h_1$ will tend, under the gradient flow, toward the minimum at $h_0$.

One can continue in this fashion to analyze the stability and instability of successive critical points
leading to a more-or-less complete picture of the geometry of the gradient flow in this instance.

\section{Gradient and Hamiltonian  flows of Strichartz functional in one dimension}
\label{sec:grad_flow}

The Strichartz inequality for linear Schr\" odinger equation in one dimension  is given by\footnote{All integrals are evaluated over the real line, unless stated otherwise. }
\begin{equation}
\int\int |e^{i t \partial_x^2} f |^6 dx dt \leq C ||f||^6_{L^2}.
\end{equation}
It is natural to consider the ratio whose supremum gives the best constant in this inequality. 
Mimicking the construction in the previous section,
we will also associate the left hand-side of the inequality with the Hamiltonian functional
\begin{equation}
H[f] = \int\int |e^{i t \partial_x^2} f |^6 dx dt.
\end{equation}
Then the ratio giving the best constant in the Strichartz inequality can be written as 
\begin{equation}\label{e:Strichartz_defn}
S[f] = \frac{H[f]}{||f||^6_{L^2}} = \frac{\int\int |e^{i t \partial_x^2} f |^6 dx dt}{(\int |f|^2 dx)^3}\ .
\end{equation}

%{\color{red}Comment(CEW) Edited the following paragraph to explain that because we don't know
%that solutions of the gradient flow in this problem are precompact, we can't apply general results
%on omega-limit-sets of gradient flows, but nonetheless regard this as a useful non-rigrous
%way of searching for best constants for inequalities.}

As in the previous section, our first  goal is to study the associated gradient flow
\begin{equation}\label{eq:gradient_2}
\dot{f} = - \nabla S[f].
\end{equation}

In this case, this gives rise to a complicated, infinite dimensional dynamical system.  Unlike
in the previous section we cannot conclude that all solutions are precompact, and so we don't know
that all initial conditions even have an omega-limit set, let alone that they will all
approach a fixed point for the flow, as is the case for the omega-limit set of solutions of
finite dimensional gradient flows.  However, we feel that searching for critical points of this
flow can still give insight into the likely candidates for the functions yielding best constants in
this type of inequalities. 
Since 
the function which gives the best constant is obviously a fixed point, one way to search for the best
constant would be look at the limit points of solutions of \eqref{eq:gradient_2}.  Of course, this 
strategy could fail if $S[f]$ has local minima other than the global minimum.  So our first
goal will be to identify critical points of \eqref{eq:gradient_2} and analyze their stability.

 As in the previous section we find it easiest to study this gradient flow
by expanding $f$ with respect to the basis of Hermite functions.
\begin{eqnarray}
f(x) = \sum_{n=0}^{\infty} \alpha_n f_n(x)\ .
\end{eqnarray}
note that because of the normalization, the denominator of the Strichartz functional has the very simple form
\begin{equation}
\left (\int |f|^2 dx \right )^3 = \left ( \sum_{n=0}^{\infty} | \alpha_n |^2  \right )^3.
\end{equation}

The other important point is that the evolution of $f$ under the free Schr\" odinger evolution is extremely simple in this basis, namely 
% \redremark{We need to double check this formula}
\begin{equation}
e^{i t \partial_x^2} f(x) = \sum_{n=1}^{\infty} \alpha_n \frac{c_n}{\sqrt{1+2it}} \left( \frac{1-2it}{1+2it} \right)^{n/2} 
H_n \left ( \frac{x}{\sqrt{1+4 t^2}} \right ) \exp \left (-\frac{x^2/2}{1+ 2it} \right ).
\end{equation}

Inserting this into the numerator of the Strichartz functional we find

\begin{eqnarray}
 \int\int |e^{i t \partial_x^2} f |^6 dx dt = \sum_{ n_1 \dots n_6} c_{n_1} c_{n_2} c_{n_2} c_{n_4} c_{n_5} c_{n_6} 
\alpha_{n_1}  \alpha_{n_2} \alpha_{n_2} \overline{\alpha_{n_4}}\overline{\alpha_{n_5}}\overline{\alpha_{n_6}}  \cdot 
\end{eqnarray}
\begin{eqnarray}
\cdot \int\int \frac{1}{(1+4 t^2)^{3/2}} \left( \frac{1-2it}{1+2it} \right)^{\frac{n_1+n_2+n_3-n_4-n_5-n_6}{2}} 
H_{n_1} H_{n_2} H_{n_3} 
H_{n_4} H_{n_5}H_{n_6}  \cdot
\end{eqnarray}
\[
\cdot \exp \left (-\frac{3 x^2}{1+4t^2} \right ) dx dt,
\]
where $H_{n_i}=H_{n_i}(\xi)$ with   $\xi = \frac{x}{\sqrt{1+4 t^2}}$.

Next, we make the change of variables $\xi = \frac{x}{\sqrt{1+4 t^2}}$ and by some miracle the space and time integrals decouple and we have:

\begin{eqnarray}
H=  \sum_{ n_1 \dots n_6} c_{n_1} c_{n_2} c_{n_2} c_{n_4} c_{n_5} c_{n_6} 
\alpha_{n_1}  \alpha_{n_2} \alpha_{n_2} \overline{\alpha_{n_4}}\overline{\alpha_{n_5}}\overline{\alpha_{n_6}} 
\end{eqnarray}
\begin{eqnarray}
\int \frac{dt}{1+4 t^2 }\left( \frac{1-2it}{1+2it} \right)^{\frac{n_1+n_2+n_3-n_4-n_5-n_6}{2}}  
\left( \int H_{n_1} H_{n_2} H_{n_3}H_{n_4} H_{n_5} H_{n_6}  e^{-3 \xi^2}  d\xi \right),
\end{eqnarray}
where $H_{n_i}=H_{n_i}(\xi)$.

%{\color{red}Comment(CEW) Rephrased the following as a lemma.}
What's more, once decoupled in this fashion, we find that the time integral can be evaluated explicitly.
Denote
\begin{equation}
\hspace{-5mm} \Lambda_{n_1,n_2,n_3,n_4,n_5,n_6} = c_{n_1} c_{n_2} c_{n_2} c_{n_4} c_{n_5} c_{n_6}  \int H_{n_1} H_{n_2} H_{n_3} H_{n_4} H_{n_5} H_{n_6}  e^{-3 \xi^2} dx.
\end{equation}

\noindent

We now have:
\begin{lemma}\label{l:time_integral}
Let  $r\neq 0$, then
\[  
\int   \frac{dT}{1+4T^2}   \left ( \frac{1-i2T}{1+i2T} \right )^{r} = \frac 1 2 \frac{\sin r\pi}{r}
\]
and if $r=0$ then the integral is equal to $\pi/2$.
\end{lemma}
\begin{proof}

\[  
\int   \frac{dT}{1+4T^2}   \left ( \frac{1-i2T}{1+i2T} \right )^{r} = \int  \frac{dT}{1+4T^2} 
\frac{e^{-ir\atan 2T}}{e^{ir \atan 2T}} = \int_{-\pi/2}^{+\pi/2} \frac{  e^{-i2rs} ds}{2\cos^2 s (1+\tan^2 s)}= 
\]
\begin{eqnarray}
\label{eq:time_int}
\frac 1 2 \int_{-\pi/2}^{+\pi/2} e^{-i2rs} ds = \frac 1 2 \frac{\sin r\pi}{r}.
\end{eqnarray}

\end{proof}

\begin{rmk} Note that integral vanishes if $r$ is a non-zero integer. \end{rmk}

Note that by parity considerations, $\Lambda_{n_1,n_2,n_3,n_4,n_5,n_6} =0 $ unless $n_1+ \dots + n_6$ is even.  This in turn means
that either $n_1+n_2+n_3$ and $n_4+n_5+n_6$ are either both even or both odd.  In either case, $n_1+n_2+n_3-n_4-n_5-n_6$ is even
and hence $\frac{n_1+n_2+n_3-n_4-n_5-n_6}{2}$ is an integer and hence by using the integral \eqref{eq:time_int}
\begin{equation}
 \int \frac{1}{(1+4 t^2) }\left( \frac{1-2it}{1+2it} \right)^{\frac{n_1+n_2+n_3-n_4-n_5-n_6}{2}} dt = 0
\end{equation}
unless $n_1+n_2+n_3-n_4-n_5-n_6=0$.  Thus, we have 

\begin{eqnarray}
\int\int |e^{i t \partial_x^2} f |^6 dx dt = \frac{2}{3} \sum_{k=0}^{\infty} \sum_{\begin{array}{c} n_1+n_2+n_3 = k \\ n_4+n_5+n_6 = k \end{array}}
\alpha_{n_1}  \alpha_{n_2} \alpha_{n_2} \overline{\alpha_{n_4}}\overline{\alpha_{n_5}}\overline{\alpha_{n_6}}  \ \Lambda_{n_1,n_2,n_3,n_4,n_5,n_6}.
\end{eqnarray}

Hence, in terms of the coefficients $\alpha_j$, we have a representation of the Strichartz functional as
\begin{equation}\label{eq:strfunc}
S[f] = \frac{\frac{2}{3} \sum_{k=0}^{\infty} \sum_{\begin{array}{c} n_1+n_2+n_3 = k \\ n_4+n_5+n_6 = k \end{array}}
\alpha_{n_1}  \alpha_{n_2} \alpha_{n_3} \overline{\alpha_{n_4}}\overline{\alpha_{n_5}}\overline{\alpha_{n_6}}  \ \Lambda_{n_1,n_2,n_3,n_4,n_5,n_6} }{( \sum_{n=0}^{\infty} | \alpha_n |^2 )^3}\ .
\end{equation}

%\begin{rmk}  Note that if our function $f$ is real, all the coefficients $\alpha_k$ will be real.  Thus, for the time being we will work in the invariant subspace of real coefficients.
%\end{rmk}

\begin{rmk} Note that this expression for the Strichartz functional is rather surprising.  In its original form
\eqref{e:Strichartz_defn}, the functional involved the entire trajectory of the function under the $f$ under the Schr\"odinger flow.  However, in \eqref{eq:strfunc}, we have reduced it to an expression involving only the spatial dependence of $f$ - the time dependence has been completely eliminated.
\end{rmk}

Using the form \eqref{eq:strfunc},  the associated gradient flow of $S[f]$ can be written as:
\begin{eqnarray}
\dot{\alpha}_{\ell} &=&  - \frac{ \partial }{\partial \overline{\alpha}_{\ell}} S[f]  \\
&=& \frac{-2  \sum_{k=0}^{\infty} \sum_{\begin{array}{c} n_1 +n_2+n_3 = k \\ n_4+n_5+ \ell = k \end{array}}
  \alpha_{n_1} \alpha_{n_2} {\alpha_{n_3}}{\overline{\alpha}_{n_4}}{\overline{\alpha}_{n_5}}  \ \Lambda_{n_1,n_2,n_3,n_4,n_5, \ell} }{( \sum_{n=0}^{\infty}
  | \alpha_n |^2 )^3} \\ 
  &+&   \frac{ 2 \alpha_{\ell}  \sum_{k=0}^{\infty} \sum_{\begin{array}{c} n_1+n_2+n_3 = k \\ n_4+n_5+n_6 = k \end{array}}
\alpha_{n_1}  \alpha_{n_2} \alpha_{n_3}{\alpha_{n_4}}{\alpha_{n_5}}{\alpha_{n_6}}  \ \Lambda_{n_1,n_2,n_3,n_4,n_5,n_6} }{( \sum_{n=0}^{\infty} | \alpha_n |^2 )^4}
\end{eqnarray}

\begin{lemma}
Every sequence $\alpha \in \ell^2$ of the form
\begin{equation}
\alpha_m = \left\{ \begin{array}{c  c} A & {\mathrm{if}} \ \ m = p* \\ 0 & {\mathrm{otherwise}} \end{array} \right.
\end{equation}
is a fixed point for the Strichartz flow.  
\end{lemma}

\begin{proof} This follows because the only way for the sums in the numerator to be non-zero is
if all the indices are equal to $p*$ and in this case, both terms vanish if $\ell \ne p*$ and they exactly cancel each other if $\ell = p*$.  
\end{proof}

\begin{rmk}  This implies that any multiple of a Hermite function is a critical point for the gradient
flow associated with the Strichartz functional.
\end{rmk}

\begin{rmk}  Another natural question is whether or not these are the only fixed points - this
would then suggest that they are the most likely candidates for yielding the best
constant in the Strichartz inequality.  So far, we haven't been able to prove that there are no other
critical points, 
though we conjecture that this is the case.
\end{rmk}

\begin{rmk}

There is an alternative dynamical formulation of the Strichartz integral in which it is interpreted 
as the Hamiltonian functional.  The equations of motion are then 
given using the familiar symplectic structure
\begin{equation}
u_t = i D_{\bar u} H.
\end{equation}

If we rewrite  Strichartz Hamiltonian by expanding $u$ in terms of the Hermite functions
as we did above, $H$ takes the form
\begin{equation}
H =  \sum_{\begin{array}{c} n_1+n_2+n_3 =  \\ n_4+n_5+n_6 \end{array}}  \Lambda_{n_1,n_2,n_3,n_4,n_5,n_6} 
\alpha_{n_1}  \alpha_{n_2} \alpha_{n_3} \overline{\alpha_{n_4}}\overline{\alpha_{n_5}}\overline{\alpha_{n_6}},
\end{equation}
and the equations of motion are given by
\begin{equation}
\dot \alpha_{\ell} =  i \frac{\partial H}{\partial \overline{\alpha}_{\ell}}.
\end{equation}

First of all, 
it is easy to see in Hermite basis  that the Strichartz Hamiltonian is invariant under the flow of the quantum harmonic oscillator discussed earlier.  In this case, the Hamiltonian is given by 
\begin{equation}
Q = \sum_{n=0}^{\infty} \, (n+\frac 1 2) \, \alpha_n \overline{\alpha}_n\ .
\end{equation}
It is also invariant if we replace $u$ by its Fourier transform, which just
multiplies the coefficients $\alpha_n$ by an $n$-dependent phase:
and under Fourier transfrom 
\begin{equation}
{\mathcal F} (\alpha_n) = e^{i\frac{\pi}{2} n} \alpha_n.
\end{equation}

Following the approach of Hani et.al., the fact that Strichartz flow commutes with the flow of quantum harmonic oscillator implies that Strichartz flow
leaves any Hermite function invariant. For the reader's convenience we give an outline of the argument from \cite{faou}. 

As the Strichartz and quantum harmonic oscillator  Hamiltonian flows commute, we can write 
\begin{equation}
e^{iLt}U(s,f) = U(s,e^{iLt} f),
\end{equation}
where $e^{iLt}f$ is the flow of quantum  harmonic oscillator with $L= \dd^2 - x^2$ and $U(s,f)$ is the Strichartz Hamiltonian flow that evolves  initial function $f$ to the new function $U(s,f)$ after time $s$. Let now, $f=f_n$, be an eigenfunction of $L$, which is a Hermite function in our particular case. Then, we have
\begin{equation}
e^{iLt} f_n = e^{i\lambda_n t}f_n \Rightarrow e^{iLt}U(s,f_n) = U(s,e^{i\lambda_n t} f_n) = e^{i\lambda_n t} U(s, f_n),
\end{equation}
where in the last equality, we used phase invariance of the Strichartz Hamiltonian flow.
Thus, we have 
\begin{equation}
e^{iLt}U(s,f_n) = e^{i\lambda_n t} U(s,f_n).
\end{equation}
Since, all eigenvalues of $L$ are simple, differentiating with respect to $t$ and setting $t=0$ we must have 
\begin{equation}
U(s,f_n) = c_n(s) f_n.
\end{equation}
Differentiating with respect to $s$ and setting $s=0$, we obtain $D_{\bar u} H(f_n) = c_n f_n$,
from which we conclude that the Hermite functions are periodic orbits for the 
Hamiltonian flow generated by the Strichartz functional. Note that this is in contrast to the case
of the gradient flow discussed earlier in this section where the Hermite functions were stationary points. 

\end{rmk}

\begin{rmk} The previous discussion of the Hamiltonian flow and its relationship to the Strichartz gradient flow assume that we are still working in one spatial dimension.  The case of higher dimensions will
be treated in a later section.
\end{rmk}

\section{Relation between constrained and unconstrained Hessians }
In this section we describe the relation between critical points corresponding to  Hermite functions in the constrained Hamiltonian and in  the  gradient flow.
While, some results can be extended to arbitrary critical points, we concentrate on those which we already know and which  will be used in the subsequent sections: Hermite functions. We also conjecture that the Hermite functions are the only critical points.

\subsection{Critical points}
We use the notation from the previous section 
\[
\alpha = (\alpha_0, \alpha_1, \alpha_2, ... ), \,\,\, \alpha_n \in {\mathbb C}.
\]
We will denote by $\alpha_k^*$ the point where $\alpha_n =0$ if $n\neq k$ and $\alpha_k \neq 0.$
We will also use real and imaginary parts of the coefficients, with $\alpha_n= p_n+i q_n$ and 
$\overline \alpha_n = p_n-iq_n$. Even though in the subsequent sections  we will mainly use real variables, some calculations in this section are more conveniently done 
in the complex variables.  Then, we restate the results in terms of the real variables.

Consider the functional  given by \eqref{eq:strfunc}
\begin{equation}
S(\alpha) = \frac{H(\alpha)}{P^3(\alpha)},
\end{equation}
where $H$ is a real-valued homogeneous polynomial  of degree 6  and $P(\alpha) = \sum |\alpha_n|^2$.
The main goal of this section is  to understand the relation between  critical points corresponding to  Hermite functions and their stability  in $S(\alpha)$ and  in $H(\alpha)$ subjected to the  constraint 
$P(\alpha) =C$. First, we observe  that both variational problems indeed have Hermite functions as critical points. 
\begin{lemma}\label{lem:constrained}
The point $\alpha_k^*$ is a critical point of $S$ if and only if $\alpha_k^*$ is a critical point of $H$ with the constraint $P(\alpha)=C.$
\end{lemma}
\begin{proof}
First, observe that for any $n\neq k$,
\begin{equation}
\frac{\partial H}{ \partial \alpha_n}(\alpha_k^*) = P(\alpha_k^*)^3 \frac{\partial S}{\partial  \alpha_n}( \alpha_k^*),
\end{equation}
since $\partial_{\alpha_j} P(\alpha_k^*)$  = 0. A similar identity holds for  $\dd/\dd \overline \alpha_n$.

Second, by invariance $S(\sigma \alpha_k^*) = S(\alpha_k^*)$ so that we have (differentiating along the real $\sigma \in \R$ and imaginary  $\sigma\in i\R$ directions at $\sigma=1$).
\begin{equation}
\alpha_k \frac{\partial S}{\partial \alpha_k}(\alpha_k^*) + \overline \alpha_k \frac{\partial S}{\partial \overline\alpha_k}(\alpha_k^*) =0
\end{equation}
and 
\begin{equation}
 i \alpha_k \frac{\partial S}{\partial \alpha_k}(\alpha_k^*) -i  \overline \alpha_k \frac{\partial S}{\partial \overline\alpha_k}(\alpha_k^*) =0,
\end{equation}
which implies $\partial_{\alpha_k} S(\alpha_k^*)= \partial_{\overline \alpha_k} S(\alpha_k^*) = 0.$ Note that we don't have to differentiate $H$ with respect to $\alpha_k$ due to the constraint, i.e. the corresponding terms do not enter the gradient.

\end{proof}
\begin{rmk}
The same conclusion (first partial derivatives vanish at $\alpha_k^*$) holds in real coordinates $(p_n,q_n)$.
\end{rmk}

\subsection{Hessians}
Now, we consider the Hessian of $S(\alpha)$ at a  critical point $\alpha_k^*$ and evaluate  partial derivatives of the second order involving at least one partial derivative 
 $\partial_{\alpha_k}$ or $\partial_{\overline \alpha_k}$.
\begin{lemma}
\[
\frac{\partial^2 S}{\partial \alpha_k \partial \alpha_n}(\alpha_k^*) = \frac{\partial^2 S}{\partial \alpha_k \partial  \overline \alpha_n}(\alpha_k^*) = 
\frac{\partial^2 S}{\partial \overline \alpha_k \partial \alpha_n}(\alpha_k^*) = 0, \,\, {\rm for} \,\, {\rm any} \,\, n.
\]
\end{lemma}
\begin{proof}
Differentiating the relation 
\begin{equation}
S(\alpha) = S(\sigma \alpha)=S(\sigma \alpha_1, \overline{\sigma \alpha}_1,  \sigma \alpha_2, \overline{\sigma \alpha}_2, ...)
\end{equation}
along the real direction $(\sigma = 1+\epsilon)$, we obtain
\begin{equation}
\sum_{n} \alpha_n \dd_{\alpha_n} S + \overline{\alpha}_n \dd_{ \overline{\alpha}_n } S =0
\end{equation}
and differentiating along the imaginary direction we get
\begin{equation}
\sum_{n} \alpha_n \dd_{\alpha_n} S  - \overline{\alpha}_n \dd_{ \overline{\alpha}_n } S.
\end{equation}
Next, differentiate both relations with respect to $\alpha_k$ and evaluate at $\alpha_k^*$:
\begin{eqnarray}
\alpha_k   \dd^2_{\alpha_k \alpha_k} S +    \overline{\alpha}_k \dd^2_{ \overline{\alpha}_k  \alpha_k} S  = 0 \\
\alpha_k   \dd^2_{\alpha_k \alpha_k} S -    \overline{\alpha}_k \dd^2_{ \overline{\alpha}_k  \alpha_k} S = 0.
\end{eqnarray}
All other terms vanish because  they either contain first partial derivatives (which vanish as $\alpha_k^*$ is a critical point) or because of $\alpha_m=0$ if $m\neq k$. 
%Setting $\alpha_k=\overline{\alpha}_k=1$, we obtain 
Since $\alpha_k \neq 0$, we immediately obtain
\be
\dd^2_{\alpha_k \alpha_k} S(\alpha_k^*) = \dd^2_{ \overline{\alpha}_k  \alpha_k} S(\alpha_k^*) = 0.
\ee
Similarly, differentiating over $\overline{\alpha}_k$, we obtain that 
\[
 \dd^2_{ \overline{\alpha}_k  \overline{\alpha}_k} S(\alpha_k^*) = 0.
\]
Next, differentiating  over $\alpha_m, m\neq k$, we obtain
\begin{eqnarray}
\alpha_k   \dd^2_{\alpha_k \alpha_m} S +    \overline{\alpha}_k \dd^2_{ \overline{\alpha}_k  \alpha_m} S  = 0 \\
\alpha_k   \dd^2_{\alpha_k \alpha_m} S -    \overline{\alpha}_k \dd^2_{ \overline{\alpha}_k  \alpha_m} S  =0.
\end{eqnarray}
Again all other terms vanish and since  $\alpha_k \neq 0$, we obtain
\be
\dd^2_{\alpha_k \alpha_m} S(\alpha_k^*) = \dd^2_{ \overline{\alpha}_k  \alpha_m} S(\alpha_k^*) = 0.
\ee
Finally, differentiating over $\overline{\alpha}_m, m\neq k$, we obtain
\be
\dd^2_{\alpha_k \overline{\alpha}_m} S(\alpha_k^*) = \dd^2_{ \overline{\alpha}_k  \overline{\alpha}_m} S(\alpha_k^*) = 0.
\ee

\end{proof}

\begin{corollary}
All second order partial derivatives in the $(p,q)$ coordinates vanish if they contain $\partial_{p_k}$ or $\partial_{q_k}$.
\end{corollary}
\begin{proof}
The calculation is straightforward using $\alpha_k = p_k+iq_k$.
\end{proof}

\begin{theorem}
Hessians evaluated at any  Hermite function of the the restricted Hamiltonian and of the gradient flow functional coincide for all second order partial derivatives that do not involve 
$\alpha_k,\overline \alpha_k$ .
\end{theorem}

\begin{rmk}
This theorem along with the above lemma imply that the Hessian corresponding to the gradient flow evaluated at a Hermite function critical point is a block matrix 
with the main block  consisting of the Hessian of the Hamiltonian and a zero block corresponding to partial derivatives involving $\alpha_k, \overline \alpha_k$   ($p_k,q_k$ in real case).
\end{rmk}

\begin{proof}
Consider now the other entries of the Hessian, which do not involve $\dd_{p_k}, \dd_{q_k}$:
\begin{eqnarray}
\frac{\dd^2}{\dd p_i \dd p_j} S = \frac{\dd^2}{\dd p_i \dd p_j} (HP^{-3}) = \dd_{p_i} ( P^{-3}\dd_{p_j} H - H 3 P^{-4} \dd_{p_j} P)=  \nonumber\\
 P^{-3}\dd_{p_j}  \dd_{p_i} H - 3P^{-4} \dd_{p_i}P \dd_{p_j} H - \dd_{p_i}H 3 P^{-4} \dd_{p_j} P +
 H 12P^{-5} \dd_{p_i} P \dd_{p_j} P - H3P^{-4} \dd_{p_j}  \dd_{p_i} P.
\end{eqnarray}

Assuming that $P(\alpha_k^*)=1$ (the calculations are similar if $P=C\neq 1$) and evaluating the above expression at $\alpha_k^*$, we obtain

\[
\frac{\dd^2}{\dd p_i \dd p_j} S(\alpha_k^*) = \frac{\dd^2}{\dd p_i \dd p_j} H(\alpha_k^*) -
6H(\alpha_k^*) \delta_{ij}.
\]

\noindent
Now, we compute the Hessian of $H(\alpha_k^*)$ restricted to the sphere $P(\alpha)=1$. Let 
\be
\hspace{-10mm} H^k = H(p_0,q_0, ..., p_k = \cos \phi \sqrt{1- \sum_{i\neq k}(p_i^2+q_i^2)}, q_k = \sin \phi \sqrt{1- \sum_{i\neq k}(p_i^2+q_i^2)}, p_{k+1}, q_{k+1}, ...),
\ee
i.e. $p_k,q_k$ variables are expressed as functions of other variables using the constraint.   
Next,
\be
\frac{\dd H^k}{\dd p_j} = \frac{\dd H}{\dd p_j}  -  \frac{\dd H}{\dd p_k} \cos \phi \frac{p_j}{\sqrt{1- \sum_{i\neq k}(p_i^2+q_i^2)}}
-\frac{\dd H}{\dd q_k}   \sin \phi \frac{q_j}{\sqrt{1- \sum_{i\neq k}(p_i^2+q_i^2)}}
\ee
%(\redremark{need better notation for differentiating restricted H wrt one variable})
and then 
\be \hspace{-5mm}
\frac{\dd^2  H^k}{\dd p_i \dd p_j} =  \frac{\dd^2  H}{\dd p_i \dd p_j} - \frac{\dd H}{\dd p_k} 
\cos \phi \frac{\delta_{ij}}{\sqrt{1-  \sum_{i\neq k}(p_i^2+q_i^2)}}  - \frac{\dd H}{\dd q_k} 
\sin \phi \frac{\delta_{ij}}{\sqrt{1-\sum_{i\neq k}(p_i^2+q_i^2)}}  +  ....
\ee
where $...$ are the remaining terms which are all multiples of  $p_s$ or $q_s$ with $s\neq k$.
Evaluating at $\alpha_k^*$, we observe that all such terms vanish and since 
\be
\cos \phi \cdot \sqrt{1 - \sum_{i\neq k}(p_i^2+q_i^2)} = p_k
\ee
 and 
 \be
 \sin \phi \cdot \sqrt{1 - \sum_{i\neq k}(p_i^2+q_i^2)} = q_k
 \ee
  we have
\be
\frac{\dd^2  H^k}{\dd p_i \dd p_j}(\alpha_k^*) =  \frac{\dd^2  H}{\dd p_i \dd p_j}(\alpha_k^*) - \delta_{ij}\frac{\dd H}{\dd p_k}(\alpha_k^*) p_k-
 \delta_{ij}\frac{\dd H}{\dd q_k}(\alpha_k^*) q_k,
\ee
where denominators $p_k^2+q_k^2=1$ when evaluated at $\alpha_k^*$.
To verify the desired equality 
\begin{eqnarray}
\frac{\dd^2 S}{\dd p_i \dd p_j} (\alpha_k^*) = \frac{\dd^2  H^k}{\dd p_i \dd p_j}(\alpha_k^*)
\end{eqnarray}
with $i\neq k, j\neq k$, we need to verify
\be
\frac{\dd H}{\dd p_k}(\alpha_k^*) p_k +\frac{\dd H}{\dd q_k}(\alpha_k^*)q_k = 6H(\alpha_k^*) .
\ee
This equality holds because the only terms contributing to both sides must  contain only 
$\alpha_k, \alpha_k^*$, which is really a single monomial  $|\alpha_k^*|^6= (p_k^2+q_k^2)^3.$ The above identity  clearly  holds for this term. 

Similarly we can verify that for $i\neq k, j\neq k$, we also have 

\be
\frac{\dd^2}{\dd q_i \dd q_j} S(\alpha_k^*) = \frac{\dd^2  H^k}{\dd q_i \dd q_j}(\alpha_k^*), \,\,
\frac{\dd^2}{\dd p_i \dd q_j} S(\alpha_k^*) = \frac{\dd^2  H^k}{\dd p_i \dd q_j}(\alpha_k^*).
\ee

\end{proof}

\section{Critical points in the one dimensional case}\label{s:oneDHessian}
Now, we compute the Hessian for the Hamiltonian case with the $L^2-$norm constraint.
To compute the Hessian, consider the second variation starting with off-diagonal terms.
\subsection{Real subspace,  Off-diagonal terms:} First we introduce some useful notation. \\

\noindent
{\bf Notation:} We will distinguish constrained derivatives from  unconstrained derivatives by using $D_S$ instead of $D$, where $S$
stands for sphere.  For example, 
\be
D^2 H[f] (h_1,h_2) 
\ee
would  denote second derivative  along  the direction $h_1,h_2$ at a point $f$ without using any constraint. The constrained derivative would be 
denoted
\be
D^2_S H[f](h_1,h_2). 
\ee

\vspace{5mm}

To compute the mixed  partial derivative of the Hamiltonian  at the critical point $f_m$, with the $L^2-$norm  constraint, let 
\be
f = f_m \sqrt{1-s_1^2-s_2^2} + s_1 f_k + s_2 f_l,
\ee
with $k\neq l$, (with the notation  $g= e^{it\Delta}f, g_m = e^{it\Delta}f_m$) and substitute in 
\be
H = \int \int |e^{it\Delta } f|^6 dx dt. 
\ee
By direct calculations, we obtain
\be
(g_m \sqrt{1-s_1^2-s_2^2} + s_1 g_k + s_2 g_l)^3 (\bar g_m \sqrt{1-s_1^2-s_2^2} + s_1\bar  g_k + s_2 \bar g_l)^3=
\ee
\be
s_1s_2 [ 9|g_m|^4 (g_k \bar g_l +c.c.)  + 6 |g_m|^2 (\bar g_m^2 g_k g_l +c.c.)] +..., \nonumber
\ee
and then
\[
 D^2_{S} H[f_m] (f_k, f_l) =  \left .  \frac{\dd^2 H}{\dd s_1 \dd s_2} \right  |_{s_1=s_2=0} (f) = 9 \int \int |e^{it\Delta} f_m|^4(   e^{it\Delta} f_k \, e^{-it\Delta}f_l + c.c.) dx dt  
\]
\be
+ 6 \int \int |e^{it\Delta} f_m|^2 (  (e^{-it\Delta} f_m)^2 e^{it\Delta} f_k \, e^{it\Delta} f_l + c.c. )  dx dt.
\ee
We need to evaluate two integrals
\begin{proposition}\label{p:hessian}
The first  integral
\[
\hspace{-5mm} I_1(k,l,m) =   \int \int |e^{it\Delta} f_m|^4  e^{it\Delta} f_k \, e^{-it\Delta}f_l  dx dt =   
\]
\be
= \delta_{kl}c_m^4 c_k c_l   \int \frac{dt}{1+4t^2} 
\int H_m^4(\xi) H_k^2(\xi) e^{-3\xi^2} d\xi.
\ee

The second integral,
\[
\hspace{-15mm} I_2(k,l,m) =  \int \int |e^{it\Delta} f_m|^2   (e^{-it\Delta} f_m)^2 e^{it\Delta} f_k \, e^{it\Delta} f_l dx dt =  
\]
\be
= c_m^4 c_k c_l
\int \frac{dt}{1+4t^2}  \int H_m^4(\xi) H_k(\xi)H_l(\xi) e^{-3\xi^2} d\xi,
\ee
if $k+l = 2m$,  and it is equal to zero otherwise. 
\end{proposition}
\begin{proof}
Straightforward  computation similar to the previous ones. 
\end{proof}
Evaluating the time integral and observing that only the second integral gives
a non-zero contribution to the off-diagonal elements, we find that nonzero off-diagonal terms are given by
\begin{eqnarray}
 D^2_{S} H[f_m] (f_k, f_l) = 12 I_2(k,l,m) =  12 \cdot \frac{\pi}{2} c_m^4 c_k c_l  \int H_m^4(\xi) H_k(\xi)H_l(\xi) e^{-3\xi^2} d\xi ,
\end{eqnarray}
where $k+l=2m$,  $ k\neq l, k\neq m, l\neq m$ and are equal to zero otherwise. We used that $\int dt/(1+4t^2) =\pi/2.$

\subsection{Real subspace,  Diagonal terms:}

Now, for $k=l$, we have $f = \sqrt{1-s^2} f_m  + s f_k$. Proceeding with similar calculations as above, we obtain
\be
\hspace{-5mm} (\sqrt{1-s^2} g_m + s g_k)^3   (\sqrt{1-s^2} \bar g_m + s \bar g_k)^3    =   s^2 (9|g_m|^4 |g_k|^2 + 3 ( |g_m|^2 g_m^2 \bar g_k^2 + c.c) - 3 |g_m|^6 )+....
\ee
As we know from Proposition \ref{p:hessian}, the second term will integrate to
zero if $k\neq m$, so diagonal terms are given by the first and the third terms
\begin{equation}\label{eq:diagonal_Gaussian_oneD}
D^2_S H[f_m] (f_k, f_k) = 2\cdot 9I_1(k,k,m) - 2\cdot 3 I_1(m,m,m)
\end{equation}
with the factor of 2 coming from differentiating twice  $s^2$.

\subsection{Imaginary subspace, Off-diagonal terms:}

If we next restrict variations to the imaginary subspace we find that the Hessian has a similar form.  Consider
variations around the critical points $f_m$ of the form:
\be
f = f_m \sqrt{1-s_1^2-s_2^2} + i s_1 f_k + i s_2 f_l.
\ee
Then
 \be
\left . \frac{\dd^2 H}{\dd s_1 \dd s_2}  \right |_{s_1=s_2=0}(f) = 9 \int \int |e^{it\Delta} f_m|^4(   e^{it\Delta} f_k \, e^{-it\Delta}f_l + c.c.) dx dt  -
\ee
\[
- 6 \int \int |e^{it\Delta} f_m|^2 (  (e^{-it\Delta} f_m)^2 e^{it\Delta} f_k \, e^{it\Delta} f_l + c.c. )  dx dt,
\]
where $k,l$ can be also equal to $m$. However, it is easy to see that if $k=m$ or $l=m$ but 
$k\neq m$ then all such terms vanish. \\

Hence, 
\be
D^2_S H [f_m] (if_k, i f_l)  = -12 I_2(k,l,m), 
\ee
where $k+l=2m$,  $ k\neq l, k\neq m, l\neq m$ and are equal to zero otherwise.

\subsection{Imaginary subspace, Diagonal terms:}
With 
\be
f = f_m \sqrt{1-s^2} + i s f_k,
\ee
we obtain
\be
 \hspace{-10mm} (\sqrt{1-s^2} g_m + is g_k)^3   (\sqrt{1-s^2} \bar g_m + is \bar g_k)^3    =   s^2 (9|g_m|^4 |g_k|^2 - 3 ( |g_m|^2 g_m^2 \bar g_k^2 + c.c) - 3 |g_m|^6 )+....
\ee
When $k=m$, we obtain zero as expected (invariance with respect to  phase rotation). For the other terms we obtain the same expressions as in the real case
\be
D^2_S H[f_m]  (i f_k, i f_k)=  2 \cdot 9 I_1(k,k,m) - 2\cdot 3 I_1(m,m,m).
\ee

\subsection{Mixed subspace, Variation in real and imaginary directions:}

For the variation in real and imaginary directions
\be
f = f_m \sqrt{1-s_1^2-s_2^2} + i s_1 f_k + s_2 f_l,
\ee
one obtains zero. Indeed,   both terms in the above expansion for second derivatives  become  $iI_1-i\bar I_1$ and $iI_2-i\bar I_2$ and  both of them vanish as $I_1,I_2$ are real.

\subsection{Structure of the Hessian restricted to the  real subspace}

In this section we consider in more detail the structure of the Hessian evaluated at $f_m$, using the form of the
matrix elements in the real and imaginary subspaces computed in the previous section.  Note that since
the off-diagonal matrix element with index $(k,l)$ is zero unless $k+l = 2m$, the 
 real part of the  Hessian consists of the two block matrices. The first one is of size  $2m\times 2m$  with 
nonzero terms only on the diagonal and anti-diagonal. We will denote this block matrix $M_{2m}$. 

The other block matrix is an infinite dimensional diagonal matrix.  Our numerics indicate that all but
possibly a finite number of the diagonal elements of this matrix are negative.
 
Regarding the Hessian restricted to the imaginary subspace, the diagonal elements are the same as in the real case while the off-diagonal elements have opposite sign. As we observe  below, this sign difference  does not affect the characteristic polynomial. 

The diagonal part of $M_{2m}$ is given by 
\be
a_k = [M_{2m}]_{kk} = 18  I_1(k, k, m) - 6 I_1(m, m, m), k =0, 1, 2, ...,  m-1, m+1, ... 2m
\ee
and anti-diagonal part is given by
\be
b_k = [M_{2m}]_{k,2m-k} = 12   I_2(k, 2m-k, m), k =0, 1, 2, ..., m-1, m+1, ..., 2m.
\ee
For example, for  $m=2$ the matrix takes the form:
\be
M_2=
\begin{bmatrix}
    a_0       &  0      &  0     &  b_0  \\
    0          &   a_1  & b_1  & 0       \\
   0            &  b_3  & a_3  & 0      \\
    b_4       & 0      &   0    &  a_4
\end{bmatrix}
\ee

The determinant of $M_{2m}$  can be factorized as follows
\be
\det M_{2m} = (a_0 a_{2m} - b_0 b_{2m})(a_1 a_{2m-1} - b_1 b_{2m-1}) ...(a_{m-1} a_{m+1}-b_{m-1} b_{m+1})
\ee
and then characteristic polynomial is given by
\be
p(\lambda)=((a_0-\lambda)(a_{2m}-\lambda) - b_0 b_{2m}) ...((a_{m-1}-\lambda) (a_{m+1}-\lambda)-b_{m-1} b_{m+1}).
\ee
\begin{rmk}
Note that anti-diagonal elements enter only in quadratic expressions. Therefore, the characteristic polynomials are essentially the same for the real and imaginary cases. 
The only difference is an extra zero eigenvalue in the imaginary case due to the variation along the  given Hermite mode $f_m$.

\end{rmk}

Since the matrix is symmetric  $(b_i = b_{2m-i})$ each quadratic polynomial has either two real roots or one double zero root.

Consider, $i-$th polynomial
\be
p_i(\lambda)=(a_i-\lambda)(a_{2m-i}-\lambda) - b_i b_{2m-i} = \lambda^2 -(a_i+a_{2m-i})\lambda + a_i a_{2m-i} - b_i b_{2m-i}
\ee
with eigenvalues given by
\be
\lambda_i^{\pm} = \frac{1}{2} ((a_i+a_{2m-i}) \pm \frac 1 2 \sqrt{(a_i+a_{2m-i})^2 - 4 (a_i a_{2m-i}-b_i b_{2m-i})}).
\ee

%{\color{red} $\lambda_+$ can become negative only if both $a_1,a_{2m}$ are negative and $a_1 a_{2m}> b_1^2$.} \\ 

In the next section  we provide some results of  numerical simulations. \\
 
\noindent
{\bf Conjecture:}
The  Hessian of the $m$-th mode restricted to real subspace has at least $m$ positive eigenvalues. \\

\subsection{Numerical Experiments}

\begin{figure}[ht]
\includegraphics[height=40mm]{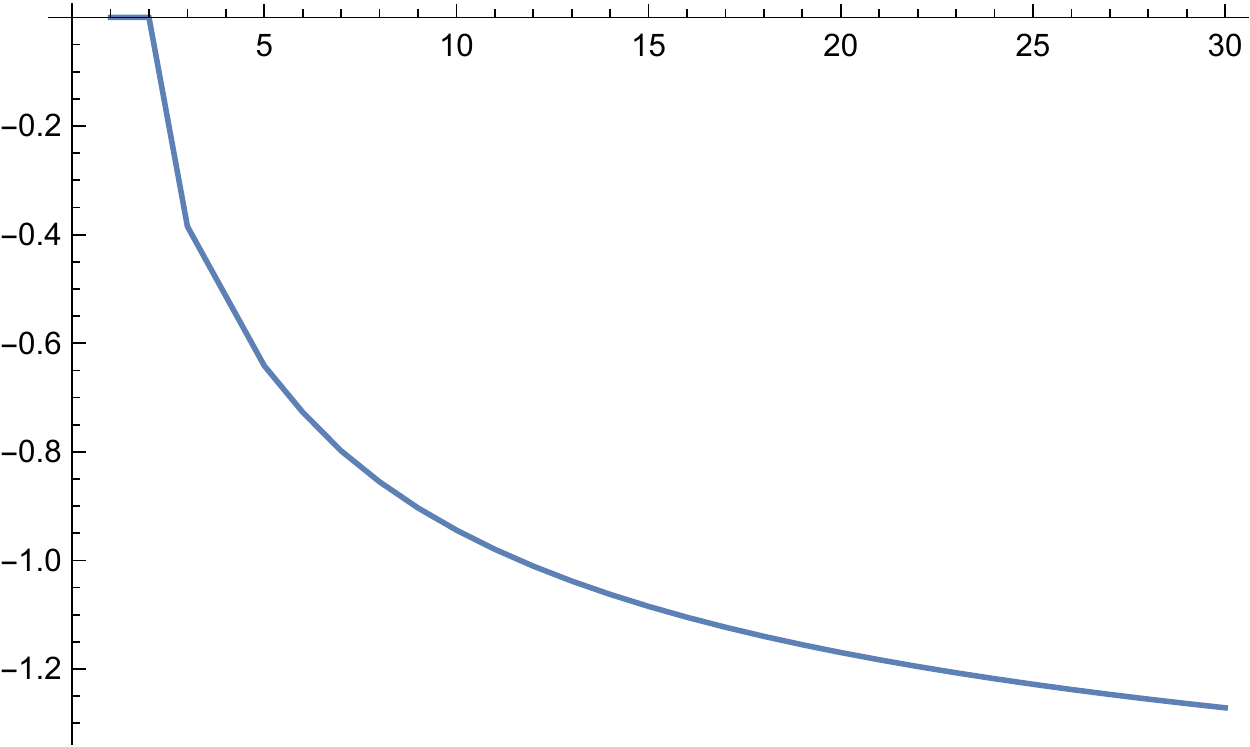}
\caption{Eigenvalues for the Gaussian.}
\label{gausseig}
\end{figure}

We use the above formulas to compute eigenvalues of the Hessian for various Hermite modes. 
\subsubsection{Gaussian: 0-th Hermite mode}\label{s:Gaussian}
First, we compute eigenvalues for Hessian matrix at the ground state mode (Gaussian). As expected, the eigenvalues are nonpositive. There are two zero eigenvalues and all other eigenvalues  are negative, as can be seen in the figure below.  In the next subsection, we demonstrate 
that these zero eigenvalues are related to symmetries of the problem, but first we consider the Hessian
matrix at critical points corresponding to higher Hermite functions.

\subsubsection{Higher modes: 1st Hermite mode}

\begin{figure}[ht]
\includegraphics[height=40mm]{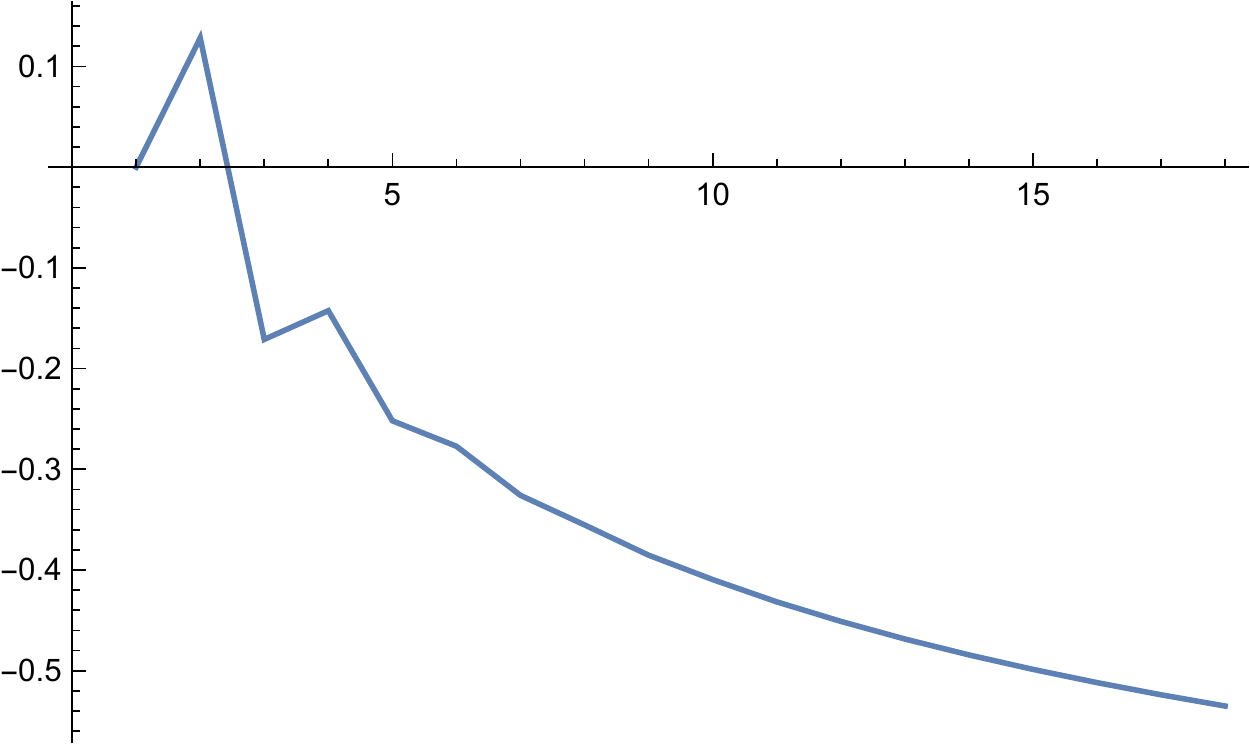}
\caption{Eigenvalues for the first mode.}
\label{1stmode}

\end{figure}

Eigenvalues from the 2 by 2 matrix are given by 
\[
-1.11022*10^{-16}, 1.1547,
\]
where the first number is interpreted as $0$.
The first few  eigenvalues of the complementary submatrix (shown on the figure \ref{1stmode}), containing only diagonal terms are given by
\[
0, 0.1283, -0.171067, -0.142556, -0.251848, -0.277191, ...
\]
with the rest of the eigenvalues appearing to be negative. This, there are 2 positive eigenvalues, 2 zero eigenvalues, with the rest being negative.

\subsubsection{Higher modes: 2nd Hermite mode}

Eigenvalues from the 4 by 4 matrix are given by
\[
1.06917, 0.299367, 2.3239*10^{-10}, 5.57755*10^{-12},
\]
where the last two numbers are interpreted as zeros.
The first 7 eigenvalues of the complementary submatrix containing only diagonal terms are given by
\[
0.114044, 0.0443506, -0.118796, -0.0533264, -0.174391, -0.153076,-0.209375, ..
\]
The next plot shows 30 eigenvalues of that submatrix.

\begin{figure}[ht]
\includegraphics[height=40mm]{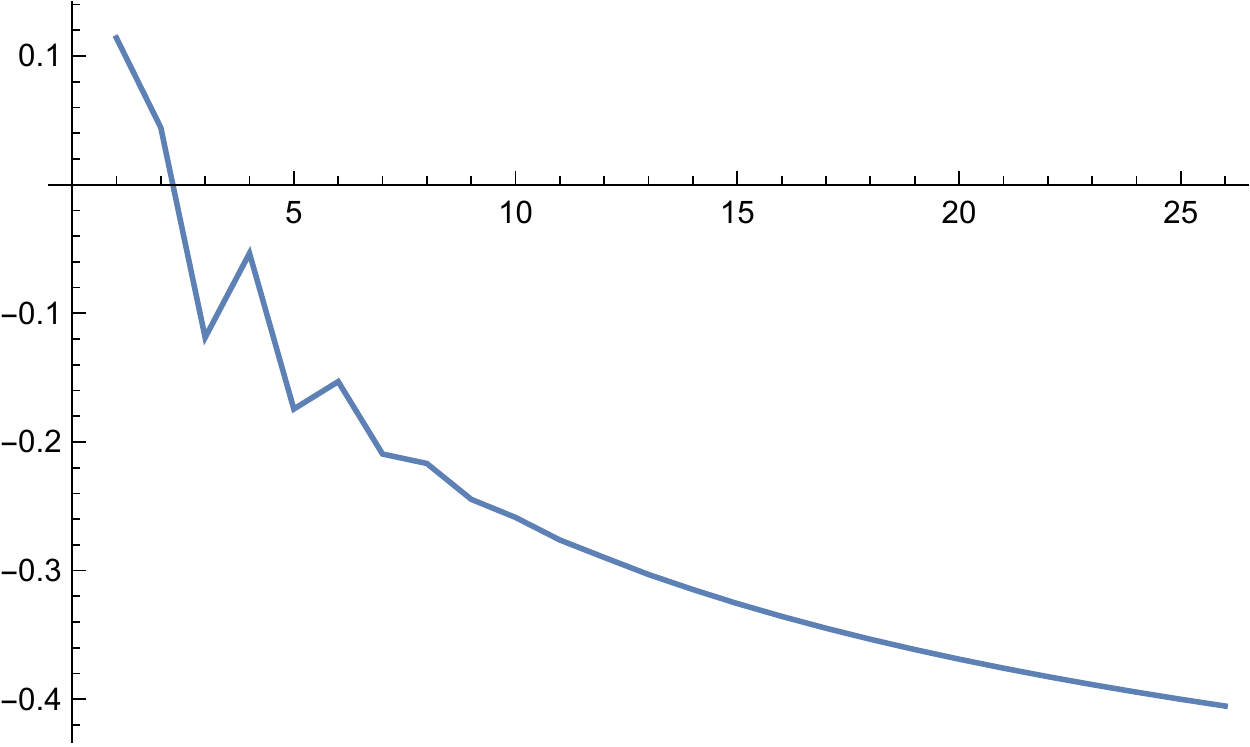}
\caption{Eigenvalues for the 2nd mode Hessian.}
\label{4thmode}
\end{figure}

Thus,  there are 4 positive eigenvalues, 2 zero eigenvalues. All other eigenvalues are negative. 

\subsubsection{Higher modes: 10-th mode Hessian}

In case $m=10$, the corresponding $2m\times 2m$  matrix has 20 eigenvalues, given below

\[
-0.0721553, -0.0607931, -0.0473447, -0.031091, -0.0134169,  -0.0107972, -0.00261104, 
\]
\[
0., 0., 0.00340212, 0.00942436, 0.01268, 0.0156644, 0.0378192, 0.0561792, 0.0731271, 
\]
\[
0.0838498, 0.149501, 0.330481, 0.654569.
\]
There are two zero eigenvalues, 7 negative eigenvalues and 11 positive eigenvalues.
The figure \ref{10mode} shows the behavior of eigenvalues corresponding to the diagonal submatrix. The numerical simulations strongly suggest that all those eigenvalues are negative.

\begin{figure}[ht] 
\includegraphics[height=40mm]{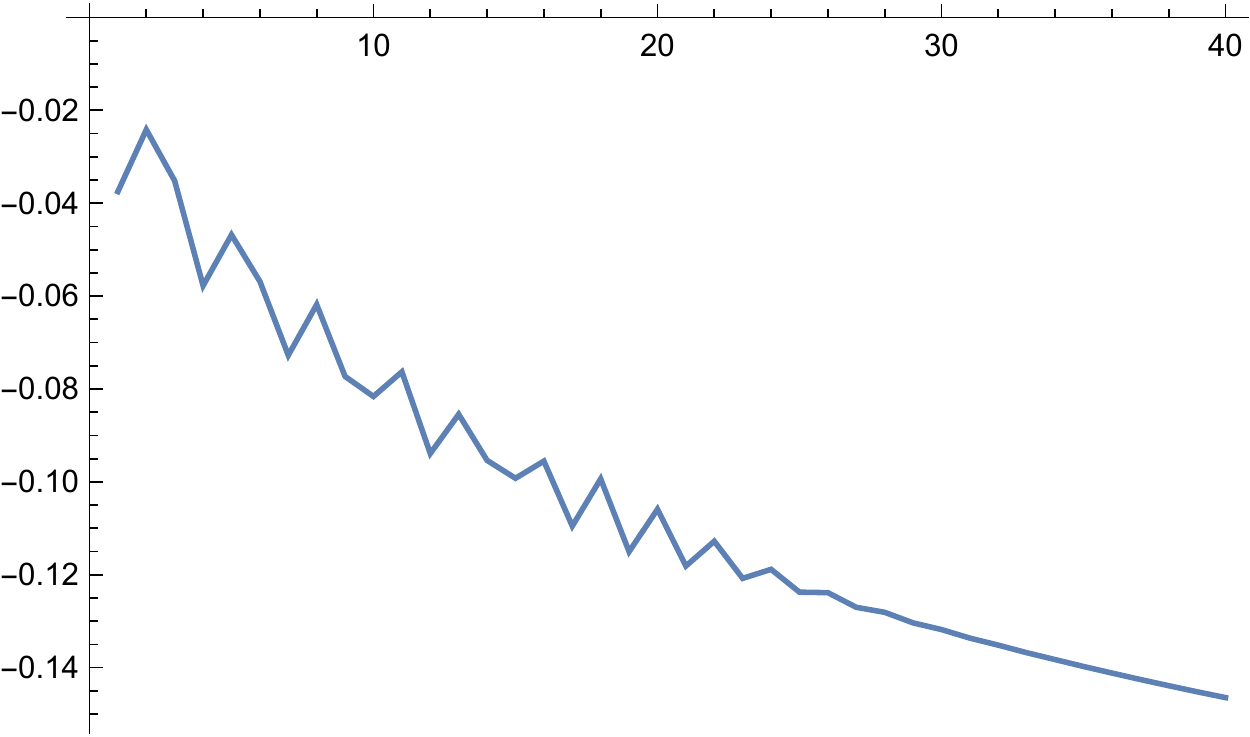}
\caption{Eigenvalues for the 10-th mode Hessian}
\label{10mode}
\end{figure}

Our numerical experiments suggest that the number of positive eigenvalues grows
in a close to linear fashion with $m$.  Recall that for the example of the quantum mechanical
harmonic oscillator which we considered explicitly in Section \ref{s:QM}, we proved that there were
exactly $m$ positive eigenvalues of the Hessian computed at the critical point $h_m$.  That
allowed us to understand the geometry of the gradient flow in that simple example in
terms of connections between the stable and unstable manifolds of various critical points.
While our understanding of the global dynamics of the gradient flow generated by the Strichartz
functional is rudimentary in comparison, these local results give at least a hint of the 
structure of this flow.  However, the increase in the number of positive eigenvalues of the
Hessian matrix at successive critical points is far less regular than in the case of the quantum
mechanical harmonic oscillator.   While our numerics (see figure \ref{onehalf} below) indicate
that as $m$ grows, the number of positive eigenvalues is approximately $m$, there is a
large variation with $m$, particularly for smaller values of $m$.  This suggests that the 
nature of the gradient flow is much more complicated than in the case of the harmonic oscillator.

\begin{figure}[ht]
\includegraphics[height=40mm]{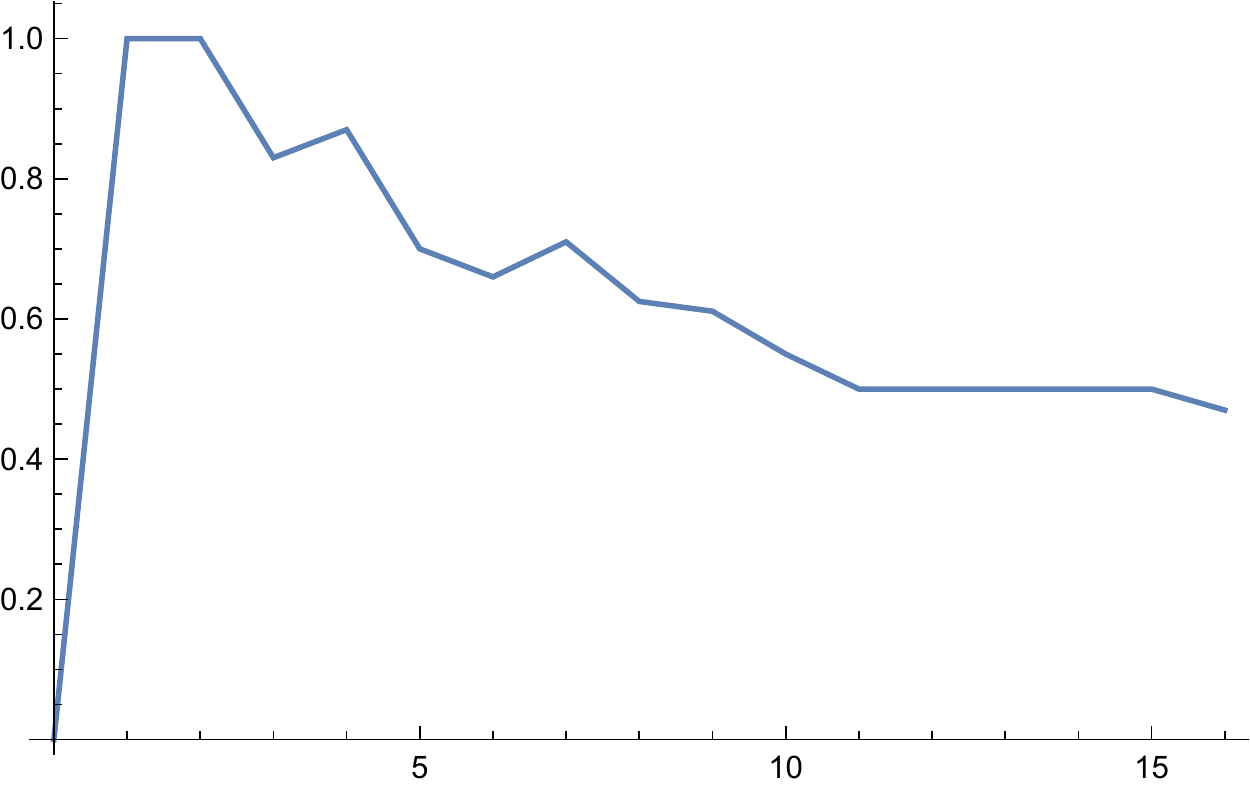}
\caption{The ratio of the number of positive eigenvalues to $2m$.}
\label{onehalf}
\end{figure}

\iffalse
\subsection{Hessian in imaginary subspace}
As found above, the Hessian restricted to the imaginary subspace consists of two block-matrices. One is 
diagonal with positive eigenvalues. The other one is $2m+1 \times 2m+1 $ which is similar in structure to
the Hessian in real subspace. 
\[
9\cdot 2\cdot I_1(k,k,m), k =0, 1, 2, ...,  m-1, m,  m+1, ... 2m
\]
and anti-diagonal part is given by
\[
-6 \cdot 2\cdot  I_2(k, 2m-k, m), k =0, 1, 2, ..., m-1, m,  m+1, ..., 2m,
\]
i.e. there is one extra entry corresponding to deformation in $if_m$ direction. Also, the signs of anti-diagonal entries are negative 
except for the new entry which is equal to the difference of $9\cdot 2 I_1(m,m,m) - 6\cdot 2 I_2(m,m,m)$ which turns out to be positive.
This, all the eigenvalues are still positive  as in the real subspace.

\fi
%Thus, to find the number of positive/negative eigenvalues, we need to compute 
%\[
%p_i(0) = a_i a_{2m-i} - b_i b_{2m-i} = 9^2 I_1(i,i,m)*I_1(2m-i,2m-i,m) - 6^2 I_2()
%\]

\subsection{Presence of zero eigenvalues due to translation invariance}
\subsubsection{Near Gaussian}
In this section we investigate the relation of zero eigenvalues to the symmetries of the variational problem.

The Hessian of the Hamiltonian computed in the previous sections  contains the matrix element 
\be
\left . \frac{d^2}{ds^2} \right |_{s=0} H[f_0\sqrt{1-s^2}+ s f_1] = D^2 H[f_0] (f_1, f_1 ) + DH[f_0] ( -f_0),
\ee
where $f_0,f_1, ...$ are normalized Hermite functions
\be
f_0= \pi^{-1/4} e^{-x^2/2}, f_1 = \pi^{-1/4} \sqrt{2} x e^{-x^2/2}, f_2 = \pi^{-1/4} (1/\sqrt{2}) (2x^2-1) e^{-x^2/2}.
\ee

 This corresponds to the second variation
about the ground state in the $h_1$ direction.  
Our computations showed that this matrix element  was  zero.  We now verify that
this zero eigenvalue results from the translation 
invariance of the Strichartz functional. Differentiate the Hamiltonian along the $x$ direction
\be
\left . \frac{d^2}{dc^2} \right |_{c=0} H(f_0(x+c)) = D^2 H[f_0] (f'_0, f'_0) + DH[f_0] (f''_0).
\ee
Direct computations show that 
\be
f'_0 = -(1/\sqrt{2})f_1, \,\,\,  f''_0 = 2f_2 - \frac 1 2 f_0.
\ee
Substitute these into the previous expression and  since $DH[f_0](f_2 ) = 0$, we can conclude 
\be
\left . \frac{d^2}{ds^2} \right |_{s=0} H[f_0\sqrt{1-s^2}+ sf_1] = \left . \frac 1 2 \frac{d^2}{dc^2} \right |_{c=0} H[f_0(x+c)] = 0. 
\ee

\subsubsection{Near Hermite functions other than the Gaussian}
Let $f_m$ be the $m-$th order Hermite function and consider
\be
\left . \frac{d^2}{dc^2}\right |_{c=0} H[f_m(x+c)] = D^2H[f_m] (f'_m, f'_m) + DH[f_m] (f''_m).
\ee
Recall  a well known Hermite functions  identity 
\be
f'_m = \sqrt{\frac m 2} f_{m-1} - \sqrt{\frac{m+1}{2}} f_{m+1}, m \geq 2
\ee
and differentiate it twice to obtain
\be
f''_m = - (m+\frac 1 2 ) f_m + f_m^{\perp},
\ee
where $f^{\perp} \in  \{g: (g,f) = 0\}$.

Introduce  normalization of $f'_m$, 
\be
f= \frac{f'_m}{\sqrt{m+1/2}}= \sqrt{\frac{m}{2m+1}} \, f_{m-1}- \sqrt{\frac{m+1}{2m+1}} \, f_{m+1},
\ee
and compute
\be
\left . \frac{d^2}{ds^2} \right |_{s=0} H [\sqrt{1-s^2} \, f_m + s \,  f] = D^2 H[f_m] (f,f) - DH[f_m] (f_m) = 
\ee
\[
 =\frac{1}{m+1/2} D^2 H [f_m] (f'_m,f'_m) - DH[f_m] (f_m) =  
 \]
 \[
 =\frac{1}{m+1/2}    D^2 H [f_m] (f'_m,f'_m) -  (m+1/2) DH[f_m] (f_m) =
 \]
 \[
=  \left . \frac{1}{m+1/2} \, \frac{d^2}{dc^2} \right  |_{c=0}H [f_m(x+c)] = 0.
\]
This strongly suggests that the function $f$ is a zero eigenvector. To prove that this is so, consider an auxiliary  function of two variables
\be
g(s_1,s_2) = H[\sqrt{1-s_1^2-s_2^2} f_m + s_1 f_{m-1} + s_2 f_{m+1}]. 
\ee
The Hessian of $g$ coincides with the central $2\times 2$ block of the $2m\times 2m$ block of the full Hessian.
On the other hand, 
\be
G(t) = g\left (\sqrt{\frac{m}{2m+1}}t, -\sqrt{\frac{m+1}{2m+1}}t \right ) = H[\sqrt{1-t^2} f_m + t f]
\ee
and we already know $G''(0) =0.$ Thus,  the quadratic form corresponding to the Hessian of $g(s_1,s_2)$ 
vanishes along the direction corresponding to $f$ and then $f$ is the zero eigenfunction.  
  
  \subsubsection{Second zero eigenvalue for variations near the Gaussian}
  
 Recall $f_0 = c_0 e^{-x^2/2}$, where $c_0^2 = 1/\sqrt{\pi}$. 
 Define
 \be
 f_c  = e^{ic(4x^2-2)} f_0
 \ee
 
 In \cite{faou}, the authors show that the Strichartz hamiltonian commutes with the flow
 generated by the quantum harmonic oscillator.  (See also discussion at the end of Section  \ref{sec:grad_flow}.)  As a consequence, we have
 $$
 H(f_c) = H(f)\ .
 $$
 
Differentiate this expression with respect to  $c$ and evaluate it at $c=0$.
\be
0= \left . \frac{d^2}{dc^2} \right |_{c=0} H[f_c] = D^2 H[f_0] (f',f') + DH[f_0] (f'' ) = 
\ee
\[
  = D^2 H[f_0] \,\, ( i (4x^2-2)f_0,   i (4x^2-2)f_0 ) +
DH[f_0] ( -(4x^2-2)^2 f_0).
\]
Note that although quadratic form $D^2 H [f] (w,w)$ has terms $ww, w\bar w, \bar w^2$, but as found in the previous sections 
the terms containing $w^2, \bar w^2$ all vanish due to orthogonality relations (assuming $(w,f)=0$). Thus, $i$ in the above expression can be taken out without changing the value.

Next, observe 
\be
(4x^2-2)^2 f_0 = \alpha f_0 + f_0^{\perp} \Rightarrow \alpha = \int (4x^2-2)^2 c_0^2 e^{-x^2} dx = 8.
\ee
Finally, note that $(4x^2-2) f_0= (c_0/c_2) f_2$. Combining these,
we obtain
\be
\hspace{-5mm}0=\left . \frac{d^2}{dc^2}\right  |_{c=0} H[f_c] = \frac{c_0^2}{c_2^2} D^2 H [f_0] (f_2,f_2) - 8 DH[f_0] f_0 = 
8(  D^2H [f_0] (f_2,f_2) -  DH [f_0] (f_0) ),
\ee
  since $c_2^2 = 1/(\sqrt{\pi} 2^2 2!)$. 
  
 The last expression is proportional to
 \be
 \left . \frac{d^2}{ds^2} \right |_{s=0} H[\sqrt{1-s^2}f_0 + s f_2] = D^2 H [f_0] (f_2,f_2) - DH [f_0] (f_0),
 \ee
 which explains the presence of the second zero eigenvalue in the Hessian evaluated at the Gaussian and restricted to real subspace.
  
  \subsubsection{Second zero eigenvalue for variations near Hermite functions other than the Gaussian}
  
 Let $\phi_m(x)$ be a quadratic function to be defined later, 
 
 \be
 f^c(x) = e^{i c \phi_m(x)} f_m, 
 \ee
 
 and consider 
 
 \be
 \left . \frac{d^2}{dc^2} \right  |_{c=0 } H[f^c] = D^2 H [f_m]\, ( i \phi_m(x) f_m,  i \phi_m(x)f_m)) + DH [f_m]  (- \phi_m(x)^2 f_m). 
 \ee
 
 We now use twice the following identity for Hermite polynomials 
 
 \be
 2xH_n = H_{n+1} + 2n H_{n-1}
  \ee
  
  to obtain
  
  \be
  4x^2 H_m = 2x (H_{m+1} + 2m H_{m-1}) =   H_{m+2} + 4m(m-1) H_{m-2} + 2(2m+1) H_m,
  \ee
  
  which implies
  
  \be
 ( 4x^2 -2(2m+1)) H_m = H_{m+2} + 4m(m-1) H_{m-2}. 
 \ee
 
 We will now choose 
 
 \be
 \phi_m(x) = ( 4x^2 -2(2m+1))
 \ee
 
  and define a  function 
  
  \[
  \tilde f_m(x)  =  \phi_m(x) f_m(x).
  \]
  
  Then we have from the above phase  invariance relation
  
  \be
  \hspace{-5mm }0 = D^2 H[f_m] (i\tilde f_m, i\tilde f_m) - DH[f_m] (\phi_m(x)^2 f_m) = D^2 H[f_m] (\tilde f_m, \tilde f_m) - DH [f_m] ( \alpha_m f_m),
  \ee
  
  where we used again  $D^2 H [f_m] (iz, iz)  = D^2 H [f_m] ( z,z) $ if $(f_m,z) =0$ and where
  
 \be
\alpha_m = (\phi_m(x)^2 f_m, f_m) = (\phi_m(x) f_m, \phi_m(x) f_m) = (\tilde f_m,\tilde f_m).
\ee

 Therefore, $\tilde f_m$ is a zero eigenvector,  because then the above expression is 
 proportional to the corresponding term in the Hessian
 
\be
\left . \frac{d^2}{ds^2}\right |_{s=0} H \left [\sqrt{1-s^2} f_m + s \frac{\tilde f_m}{\sqrt{\alpha_m}}\right ] =0.
\ee

%\subsubsection{Evaluation of space integral}
%
%
%
%We will use product formula in a symbolic form
%\[
%H_n(\gamma x) = \sum_{0}^{\lfloor{\frac n 2}\rfloor} \sigma_{k}^n H_{n-2k}(x)
%\]
%to find out which terms vanish because of orthogonality.
%
%Consider 
%\[
%\int  \exp \left (-q \xi_j^2/2\right )  H_{k_j}(\xi_j)  H_{l_j} (\xi_j) d\xi_j =
%\]
%Clearly integral vanishes if $k_j, l_j$ are not both odd or even, otherwise  
%\[
%= \gamma \int e^{-\eta_j^2} H_{k_j}(\gamma \eta_j) H_{l_j}(\gamma \eta_j) d\eta_j = 
%\gamma \int e^{-\eta_j^2} \sum_{k,s} \sigma_{r}^{k_j}\sigma_{s}^{l_j} H_{k_j-2r}(\eta_j)H_{l_j-2s}(\eta_j)d\eta_j =
%\]
%\[
%= \gamma \sum_{k_j-2r = l_j-2s \atop 0 \leq r \leq  \lfloor{k_j/2} \rfloor }    \sigma_{r}^{k_j}\sigma_{s}^{l_j}\int e^{-\eta_j^2} H^2_{k_j-2r}(\eta_j)d\eta_j.
%\]

\section{High-dimensional Strichartz functional}

In dimension $d$, the functional whose critical points we are seeking takes the form,
\be
H(u) = \int_{\R^d} \int_{\R^1} |e^{it\Delta} u|^q dx dt,
\ee
subject to the $L^2$ norm constraint $||u||_{L^2}=C$, where $q=(4/d) +2$. This functional is bounded in $L^2$ which is equivalent to the  Strichartz inequality. We claim that this functional is invariant under Fourier transform in $\R^d$ for any $d$.   We discussed this fact in dimension $d=1$ in Section \ref{sec:grad_flow}, and
it can also be shown by direct calculations in $d=2$. By using a slightly different approach, we get a simple proof  of this fact for any $d$.

\subsection{Convenient representation of Strichartz integral}
Recall that the free Schr\" odinger evolution can be written
\be
e^{it\Delta} u = \frac{1}{(4 \pi i t)^{d/2}}\int_{\R^d} e^{\frac{i|x-y|^2}{4t}} u(y) dy,
\ee
where $x,y \in \R^d$ and $|x|$ is Euclidean norm in $\R^d$. We will denote by $(x,y)$ the inner product in $\R^d$.
Substitute the last expression  in the Strichartz integral to obtain
\begin{eqnarray}
H(u) = \frac{1}{( 4\pi)^{qd/2}}  \int_{\R^1} \int_{\R^d}  \frac{1}{|t|^{qd/2}} \left  |\int_{\R^d} e^{\frac{-i(x,y)}{2t}} e^{\frac{i |y|^2}{4t}} u(y)dy \right |^q dx dt.
\label{eq:ham}
\end{eqnarray}

Now, make the change of variables in the integral
\[
x= \zeta/2\tau, \,\,\, t = 1/4 \tau, \,\, \zeta\in \R^d, x\in \R^d.
\]
The Jacobian of this transformation is: $dxdt = \frac{1}{2^{d+2}|\tau|^{d+2}}d\zeta d\tau$,
so we have
\be
H(u) = \frac{1}{2^{d+2}}   \frac{4^{qd/2}}{(4\pi)^{qd/2}}  \int_{\R^d} \int_{\R^1}  \left  |\int_{\R^d} e^{-i (\zeta,y)} e^{i \tau |y|^2} u(y)dy \right |^q d\zeta  d\tau,
\ee
and then
\be
H(u) =   \frac{1}{(2\pi)^{d+2}}  \int_{\R^d} \int_{\R^1}  \left  |\int_{\R^d} e^{-i (\zeta, y)} e^{i \tau |y|^2} u(y)dy \right |^q d\zeta  d\tau.
\ee

\subsection{Fourier transform}

Now, recall that the Fourier transform in $\R^d$ is defined as:
\be
{\mathcal F}(u) = \frac{1}{(2\pi)^{d/2}} \int_{\R^d} e^{i (y, z)} u(z) dz.
\ee

We have 
\begin{eqnarray}\label{e:Fourier}
&& H({\mathcal F}(u) )=   \\ \nonumber
&& \qquad  = \frac{1}{(2\pi)^{d+2}} \frac{1}{(2\pi)^{qd/2}} \int_{\R^d} \int_{\R^1}  \left  |\int_{\R^d} \int_{\R^d} e^{-i (\zeta,  y)} e^{i \tau |y|^2} 
  e^{i (y, z)} u(z) dz dy \right |^q d\zeta  d\tau,
\end{eqnarray}

Now, evaluate the integral over $y$,  inside  $|*|$: 
\be
\int_{\R^d} e^{-i (\zeta, y)} e^{i \tau |y|^2}   e^{i (y,z)}  dy = \int_{\R^d}  e^{i\tau |y+\frac{z-\zeta}{2\tau}|^2} e^{-i\tau \frac{|z-\zeta|^2}{4\tau^2}} dy =
\frac{K}{\tau^{d/2}} e^{-i \frac{|z-\zeta|^2}{4\tau}},
\ee
where $K= (1+i)^{d}(\pi/2)^{d/2}$. Note that $|K|= \pi^{d/2}$.\\

\begin{rmk}  Note that the integral in the previous equality is not absolutely convergent and hence
the interchange of the order of the $z$ and $y$ integrals in \eqref{e:Fourier} is not justified
by Fubini's theorem.  We can get around this problem by a standard trick of rewriting
\begin{eqnarray}
&& \int_{\R^d} \int_{\R^d} e^{-i (\zeta,  y)} e^{i \tau |y|^2} 
  e^{i (y, z)} u(z) dz dy \ \\  \nonumber
  && \quad =   \lim_{\epsilon \to 0} \int_{\R^d} \int_{\R^d} e^{-\epsilon |y|^2}
   e^{-i (\zeta,  y)} e^{i \tau |y|^2} 
  e^{i (y, z)} u(z) dz dy \ \\ \nonumber
 && \qquad \qquad  =  \lim_{\epsilon \to 0} \int_{\R^d} \int_{\R^d} e^{-\epsilon |y|^2}
   e^{-i (\zeta,  y)} e^{i \tau |y|^2} 
  e^{i (y, z)} u(z) dy dz
\end{eqnarray}
and then proceeding to evaluate the integral over $y$ as above, taking the limit $\epsilon \to 0$
after evaluating the integral.  This leads to the same result as the computation above.
\end{rmk}

Finally, we obtain
\be
H({\mathcal F}(u)) =    \frac{\pi^{d+2}}{(2\pi)^{d+2} (2\pi)^{qd/2}}   \int_{\R^d} \int_{\R^1} \frac{1}{|\tau|^{qd/2}}  \left  |\int_{\R^d} \int_{\R^d}   e^{-i \frac{|z-\zeta|^2}{4\tau}} u(z) dz \right |^q d\zeta  d\tau,
\ee
or equivalently
\be
H({\mathcal F}(u)) =    \frac{1}{ (4\pi)^{qd/2}}   \int_{\R^d} \int_{\R^1} \frac{1}{|\tau|^{qd/2}}  \left  |\int_{\R^d} \int_{\R^d}   e^{-i \frac{|z|^2}{4\tau}} e^{i \frac{ (z,\zeta)}{2\tau}} u(z) dz \right |^q d\zeta  d\tau,
\ee

which is equal to \eqref{eq:ham}.  The exponents inside the integral have the wrong signs but it is easy to check that it does not affect the value.

\begin{rmk}
In Section \ref{sec:grad_flow},  we showed that the flow generated by the quantum mechanical oscillator
commutes with the Hamiltonian flow generated by the Strichartz functional in one dimension.
This had previously been proven in dimension two by  Faou et. al. in \cite{faou}.
By extending their argument, one  can show that the Hamiltonian flow commutes 
with the flow of quantum harmonic oscillator in all dimensions.  
\be
\{ H, |\nabla u|^2 + |x|^2 |u|^2 \}  = 0.
\ee
This also means that these operators share the same eigenspaces.  However, we won't use
that result in what follows, so we don't pursue this point further.

\end{rmk}

\section{Local structure of the Strichartz functional near Gaussian}
The goal of this section  is to study the Strichartz functional in the vicinity of the Gaussian. Recall that in dimension 3 and higher it is unknown if Gaussian is a minimizer.   In this section we first prove that the Gaussian is a critical point of the Strichartz gradient flow in any dimension, and then we present evidence, partly numerical and partly theoretical, that it is at least a local minimizer.
Recall that by Lemma \ref{lem:constrained}, a function is a critical point of the Strichartz 
gradient flow if and only if it is a critical point of the Strichartz Hamiltonian $H$, under
variations which conserve the $L^2$ norm.

\subsection{First variation}
Here, we verify that the first variation of the Strichartz Hamiltonian vanishes at the Gaussian under
variations that conserve norm.
We denote by $f_k$ normalized Hermite
functions in dimension $d$
\be
f_k(x) = c_k H_k(x) e^{-|x|^2/2} = c_{k_1 ... k_d}H_{k_1}(x_1) \cdots H_{k_d}(x_d) \, e^{-\frac 1 2 (x_1^2 + \cdots + x_d^2)}.
\ee

Let
\be
f(s) = f_0 \sqrt{1-s^2} + f_k s \,\, {\rm and } \,\, 
g(s,t) = e^{it\Delta} f(s)
\ee
and compute 
\[
\left . \frac{d}{ds}  \right  |_{s=0} H(f(s)). 
\]

Note first that 
\be
\frac{d}{ds} |g|^q = \frac{d}{ds} \, (g^{q/2} \bar g^{q/2} ) = \frac{q}{2} \left (\frac{\partial_s g}{g} + \frac{\overline{ \partial_s g} }{\bar g}  \right ) |g|^q,
 \ee
then, we have
\[
\frac{d}{ds} H(f(s))|_{s=0}  = \int \int \frac{d}{ds} |g(t,s)|^q|_{s=0} dx dt = 
\frac{q}{2} \int \int  \left (\frac{(\partial_s g)(t,0)}{g(t,0)} + \frac{\overline{(\partial_s g)}(t,0)}{\bar g(t,0)}  \right ) |g(t,0)|^q dx dt =
\]
\be
=\frac{q}{2} \int \int  \left (\frac{e^{it\Delta} f_k}{e^{it\Delta}f_0} + c.c.  \right ) |e^{it\Delta} f_0|^q dx dt.
\ee
Recall
\be
e^{it\Delta} f_k = \frac{1}{(1+i 2 t)^{d/2}}  \prod_{j=1}^d \left ( \frac{1-i2t}{1+i2t} \right )^{k_j/2} 
c_{k_j}H_{k_j} \left (\frac{x_j}{\sqrt{1+4t^2}} \right )
\exp \left (\frac{-|x|^2/2}{1+i2t} \right )
\ee
with
\be
(e^{it\Delta} f_0) = \frac{c_0}{(1+i 2 t)^{d/2}}  \exp \left (\frac{-|x|^2/2}{1+i2t} \right ),
\ee
where $k=(k_1, k_2, ..., k_d)$ and $|k|= \sum k_j$ and $c_j$ are normalizing constants.

Then,
\be
\frac{e^{it\Delta} f_k}{e^{it\Delta}f_0} = \frac{c_{k_1} c_{k_2} ... c_{k_d}}{c_0} \left ( \frac{1-i2t}{1+i2t} \right )^{|k|/2} 
\prod_{j=1}^d  H_{k_j} \left (\frac{x_j}{\sqrt{1+4t^2}} \right )
\ee
and
\be
|e^{it\Delta} f_0|^q = \frac{c_0^q}{|1+ 4t^2|^{qd/4}}  \exp \left (\frac{-q|x|^2/2}{1+4t^2} \right ).
\ee

%{\color{red} Comment(CEW) Reworded the following three displayed equations, and then made more substantial changes to the end of the subsection to eliminate the discussion of the product formula - I have denoted this new material in blue.}

We now prove that the first variation vanishes at the Gaussian. Ignoring insignificant
constants, the first derivative takes the form
\be
\hspace{-10mm} \left . \frac{d}{ds} H(f(s)) \right |_{s=0} =    2 \Re \int \int \frac{dx dt}{|1+ 4t^2|^{qd/4}}  \exp \left (\frac{-q|x|^2/2}{1+4t^2} \right ) \times
\left ( \frac{1-i2t}{1+i2t} \right )^{|k|/2} 
\prod_{j=1}^d  H_{k_j} \left (\frac{x_j}{\sqrt{1+4t^2}} \right ).
\ee
Now make the change of variables introduced in Section \ref{sec:grad_flow}
to separate the time and space integrals:

\[
\xi_j = x_j/\sqrt{1+4t^2},  \,\,\, T=t,
\]
This gives
\begin{eqnarray}\nonumber 
\left . \frac{d}{ds} H(f(s)) \right |_{s=0}  &=&  2 \Re \int  \left ( \frac{1-i2T}{1+i2T} \right )^{|k|/2} \frac{(1+4T^2)^{d/2} }{(1+ 4T^2)^{qd/4}}  dT \int  \exp{(-q |\xi|^2/2)} \prod_{j=1}^d  H_{k_j} (\xi_j) d\xi \\
& = & 
2 \Re \int  \left ( \frac{1-i2T}{1+i2T} \right )^{|k|/2} \frac{dT}{1+4T^2} \,\,\, \prod_{j=1}^d  \int e^{-q\xi_j^2/2}H_{k_j} (\xi_j) d\xi.
\end{eqnarray}

First note that by construction $k \ne 0$.   The time integral vanishes if $|k|$ is even and nonzero by Lemma \ref{l:time_integral}.  On the other hand, if $|k|$ is odd, at least one $k_j$ is odd, but then the corresponding  space integral will vanish by symmetry.  Thus, we have demonstrated that
the first variation of the Strichartz Hamiltonian vanishes at the Gaussian in any dimension.

\subsection{Second variation}
\subsubsection{Off-diagonal terms in the subspace of real variations.}
Let 
\be
f(s) = f_0 \sqrt{1-s_1^2 -s_2^2} + f_{k} s_1 + f_{l} s_2, \,\,\, g(s,t) = e^{it\Delta} f(s)
\ee
be the deformation of Gaussian in the direction of the Hermite
functions $f_k$, $f_l$,  with $k\neq l$ and let $g$ be the corresponding  Schr\" odinger evolution. 
\begin{rmk}
 We will also need to compute the variation in all the directions in complex space, i.e. 
\begin{eqnarray}
f(s) = f_0 \sqrt{1-s_1^2 -s_2^2} + i f_{k} s_1 + i f_{l} s_2  \nonumber \\
f(s) = f_0 \sqrt{1-s_1^2 -s_2^2} + i f_{k} s_1 + f_{l} s_2 \\
f(s) = f_0 \sqrt{1-s_1^2 -s_2^2} + f_{k} s_1 + i f_{l} s_2, \nonumber
 \end{eqnarray}
 including the Gaussian $if_0$.
 We will see that mixed derivatives (corresponding to the 2nd and 3rd lines above) vanish and that
 variations in the purely imaginary subspace (1st line) are essentially the same as the real one.
\end{rmk}

We want to compute 
\be
\left . \frac{\partial^2 H(f(s))}{\partial  s_1 \partial s_2} \right |_{s_1=s_2=0}
\ee
but first for convenience we  evaluate
\[
\frac{\partial^2 }{\partial s_1 \partial s_2} (g^{q/2} \bar g^{q/2}) = \frac q 2 \partial_{s_2} \left (|g|^q   \left (\frac{\partial_{s_1} g}{g} +  \frac{\partial_{s_1} \bar g}{\bar g} \right )\right ) = \frac{q^2}{4} |g|^q   \left (\frac{\partial_{s_1} g}{g} +  \frac{\partial_{s_1} \bar g}{\bar g} \right )   \left (\frac{\partial_{s_2} g}{g} +  \frac{\partial_{s_2} \bar g}{\bar g} \right ) +
\]
\be
\frac q 2 |g|^q \partial_{s_2}   \left (\frac{\partial_{s_1} g}{g} +  \frac{\partial_{s_1} \bar g}{\bar g} \right ).
\ee
Note that $\partial_{s_1} \partial_{s_2} g(s)|_{s_1,s_2=0} = 0$, therefore we only need to keep terms where $g$ is differentiated once, so that
\[
\partial_{s_2}   \left (\frac{\partial_{s_1} g}{g} +  \frac{\partial_{s_1} \bar g}{\bar g} \right )= -\frac{\partial_{s_1} g \partial_{s_2} g}{g^2} - \frac{\partial_{s_1} \bar g \partial_{s_2} \bar g}{\bar g^2} + ...
\]
Next, evaluating at $s_1=s_2=0$ and integrating,  we obtain
\begin{eqnarray}\label{eq:second_var}
&& \left . \frac{\partial^2 H(f(s))}{\partial s_1 \partial s_2} \right |_{s_1=s_2=0} = \int \int \frac{\partial^2 }{\partial s_1 \partial s_2} (g^{q/2} \bar g^{q/2})(0) dx dt = \\ \nonumber
&& \qquad  
\int \int \frac{q^2}{4} |e^{it\Delta} f_0|^q \left (\frac{e^{it\Delta}f_{k}}{ e^{it\Delta}f_0}+c.c.\right )
\left (\frac{ e^{it\Delta}  f_{l}}{ e^{it\Delta} f_0}+c.c. \right )-
\frac q 2  |e^{it\Delta} f_0|^q  \left (\frac{ e^{it\Delta}   f_{k} e^{it\Delta}  f_{l}}{e^{it\Delta} f_0  e^{it\Delta} f_0 } + c.c. \right ) dx dt.
\end{eqnarray}

\vspace{5mm}

Note that all six of the terms that survive after we set $s_1=s_2=0$ are of one of the two types
that appear in the following proposition (or else a complex conjugate of one of these two.)

\begin{proposition}
For any $(k,l) \neq (0,0)$ 
\be
I^+(k,l,q) = \int \int  |e^{it\Delta} f_0|^q  \,\, \frac{e^{it\Delta}f_{k}}{ e^{it\Delta}f_0}  \, \frac{ e^{it\Delta}  f_{l}}{ e^{it\Delta} f_0} \, dx dt = 0.
\ee
For $|k|=|l|$
\be
\hspace{-5mm} I^-(k,l,q)=\int \int  |e^{it\Delta} f_0|^q  \,\, \frac{e^{it\Delta}f_{k}}{ e^{it\Delta}f_0}  \, \frac{ e^{-it\Delta}  f_{l}}{ e^{-it\Delta} f_0} \, dx dt = 
\frac{\pi}{2} c_0^{qd}  \frac{c_k c_l}{c_0^{2d}} \prod_{j=1}^d \int  \exp \left (-q \xi_j^2/2\right )
  H_{k_j}(\xi_j)  H_{l_j} (\xi_j) d\xi_j,
\ee
and $I^-(k,l,q)=0$ otherwise.
\end{proposition}
\begin{proof} Consider the first integral:

\begin{eqnarray} \nonumber
&& I^+(k,l,q) = \int \int  |e^{it\Delta} f_0|^q  \,\, \frac{e^{it\Delta}f_{k}}{ e^{it\Delta}f_0}  \, \frac{ e^{it\Delta}  f_{l}}{ e^{it\Delta} f_0} \, dx dt \\ \nonumber  && \quad
= \int \int \frac{c_0^{qd}}{|1+ 4t^2|^{qd/4}}  \exp \left (\frac{-q|x|^2/2}{1+4t^2} \right )  \frac{c_{k_1} c_{k_2} ... c_{k_d}}{c_0^d} \left ( \frac{1-i2t}{1+i2t} \right )^{|k|/2} 
\prod_{j=1}^d  H_{k_j} \left (\frac{x}{\sqrt{1+4t^2}} \right ) \\ 
&& \qquad \qquad  \qquad 
 \times  \frac{c_{l_1} c_{l_2} ... c_{l_d}}{c_0^d} \left ( \frac{1-i2t}{1+i2t} \right )^{|l|/2} 
\prod_{j=1}^d  H_{l_j} \left (\frac{x}{\sqrt{1+4t^2}} \right ) \ .
\end{eqnarray}

If we denote $c_k=c_{k_1} \dots c_{k_d}$, $c_{l}= c_{l_1}\dots c_{l_d}$ and make the
same change of variables used above to separate the time and space integrals, we obtain.
\be
I^+(k,l,q) = c_0^{qd} \frac{c_k c_l}{c_0^{2d}} \int  \frac{dT}{1+4T^2}   \left ( \frac{1-i2T}{1+i2T} \right )^{(|k|+|l|)/2}  \times \prod_{j=1}^d \int  \exp \left (-q \xi_j^2/2\right )
  H_{k_j}(\xi_j)  H_{l_j} (\xi_j) d\xi_j.
\ee
Once again, we note that the integral over $\xi$ will vanish unless all $k_j$ and $l_j$  have the same parity.
But then $(|k|+|l|)/2$ is an integer and the temporal integral vanishes by Lemma \ref{l:time_integral}.

Now we consider the second integral.
Decoupling space and time as above, one can rewrite this integral as follows:
\be
I^-(k,l,q)= \int \int  |e^{it\Delta} f_0|^q  \,\, \frac{e^{it\Delta}f_{k}}{ e^{it\Delta}f_0}  \, \frac{ e^{-it\Delta}  f_{l}}{ e^{-it\Delta} f_0} \, dx dt = 
\ee
\[ = c_0^{qd} \frac{c_k c_l}{c_0^{2d}} \int  \frac{dT}{1+4T^2}   \left ( \frac{1-i2T}{1+i2T} \right )^{(|k|-|l|)/2}  \times \prod_{j=1}^d \int  \exp \left (-q \xi_j^2/2\right )
  H_{k_j}(\xi_j)  H_{l_j} (\xi_j) d\xi_j.
\]
The integral is real valued since the integrand in the time integral is transformed
into its complex conjugate if $T$ changes sign.  The space integrals vanish  if 
at least one pair of $k_j,l_j$ have different parity. If all of them have the same parity, then 
$|k|-|l|$ is even and the time integral vanishes unless $|k|=|l|.$
Therefore,
\begin{equation}
I^-(k,l,q) = \frac{\pi}{2} c_0^{qd}\cdot  \frac{c_k c_l}{c_0^{2d}} \prod_{j=1}^d \int  \exp \left (-q \xi_j^2/2\right )
  H_{k_j}(\xi_j)  H_{l_j} (\xi_j) d\xi_j, 
\end{equation}
if $|k|-|l|=0$, otherwise $I^-(k,l,q)=0$.  
\end{proof}

Finally, using these to reexpress the second variation
integral in \eqref{eq:second_var}, we find that the off-diagonal matrix elements in the real subspace 
satisfy
\be
\frac{\partial^2 H(f(0))}{\partial s_1 \partial s_2} = 
 \frac{q^2}{4} (I^+ + I^- + \bar I^- + \bar I^+) - \frac q 2 (I^+ + \bar I^+) =
 \frac{q^2}{2} I^-(k,l,q).
\ee

\subsubsection{Diagonal terms of the Hessian, restricted to the subspace of real variations.}

Now let $g(s) = \sqrt{1-s^2} g_0 + s g_k$. Then $g'(0) = g_k$ and $g''(0) = -g_0$. Now compute
\be
\frac{d^2}{ds^2} |g|^q = \frac q 2 \frac{d}{ds} \left (|g|^q \left  (\frac{ g'}{g}+ \frac{ \bar g'}{\bar g}\right )\right ) = \frac{q^2}{4}|g|^q  \left (\frac{ g'}{g}+ \frac{ \bar g'}{\bar g}\right )^2 + \frac q 2 |g|^q \left (\frac{ g''}{g} - \frac{g'  g'}{g^2} + c.c. \right ).
\ee
Then 
\begin{eqnarray}
&& \left . \frac{d^2 }{d s^2}  H(f(s)) \right |_{s_1=s_2=0}  = \frac{q^2}{4} \int \int |e^{it\Delta} f_0|^q  \left (  \frac{e^{it\Delta} f_k}{e^{it\Delta} f_0} + c.c.\right )^2 dx dt \\ \nonumber && \qquad \qquad \qquad 
- \frac q 2 \int \int |e^{it\Delta} f_0|^q \left ( 1 + \left (\frac{e^{it\Delta} f_k}{e^{it\Delta}f_0} \right )^2  + c.c. \right ) dx dt.
\end{eqnarray}
As in the off-diagonal terms, contributions proportional to $I^+(k,k,q)$ vanish, and we are left with
\begin{eqnarray}\nonumber
&& \left . \frac{d^2 }{d s^2} H(f(s)) \right |_{s_1=s_2=0}  = \frac{q^2}{2}\int \int  |e^{it\Delta} f_0|^q \left | \frac{e^{it\Delta} f_k}{e^{it\Delta} f_0} \right |^2 dx dt - 
q \int \int |e^{it\Delta} f_0|^q dx dt \\ && \qquad \qquad \qquad  = \frac{q^2}{2} I^-(k,k,q) -q I^- (0,0,q).
\end{eqnarray}
Thus, we obtain Hessian restricted to real subspace
\begin{equation}\label{eq:Hessian_diagonal_dimd}
{\mathfrak{H}}^{\mathbb R}_{kl} = \frac{q^2}{2} I^-(k,l,q)- \delta_{kl} q I^-(0,0,q), k \geq 1, l\geq 1,
\end{equation}
where $k= (k_1, k_2, ..., k_d), l=(l_1,l_2, ..., l_d), \delta_{kl}= \delta_{k_1 l_1}  \delta_{k_2l_2}... \delta_{k_d l_d}$.
The first matrix is positive definite as it can be represented as a Gram matrix (see below). The second matrix  is diagonal proportional to the identity matrix. In the dimensions one and two, 
we already know that the full matrix is nonpositive and we expect that the same is true in higher dimensions. \\
\begin{rmk}
As a quick check, we relate these calculations of the Hessian at the Gaussian in arbitrary dimension, to the specifically one-dimensional calculations of Section \ref{s:oneDHessian}.  Note that in general dimensions, we have off-diagonal, $(k,l)$ entry in the Hessian is non-zero only if $|k|=|l|$.  (Recall that $k$ and $l$ are d-dimensional vectors with non-negative, integer entries.)  However, in one-dimension, there are no off-diagonal entries of this type and this is in agreement with our calculation that showed that the Hessian was diagonal in this case.  Turning to the diagonal entries, recall that in one-dimension, $q=6$.  From equation \eqref{eq:diagonal_Gaussian_oneD}, we found that the second variation about the Gaussian in the (real) direction $f_k$, was given by
$$
2\cdot 9I_1(k,k,0) - 2\cdot 3 I_1(0,0,0)
$$
Comparing the definitions of $I_1$ and $I_2$ with the definition of $I^{\pm}$, this
becomes
$$
18 I^{-}(k,k,6) - 6 I^{-}(0,0,6)
$$
which agrees with the expression in \eqref{eq:Hessian_diagonal_dimd}.
\end{rmk}

%{\color{red}{\tt Comment (CEW):  Moved the subsection on the ``Structure of the Hessian'' to after the discussion of the Hessian in complex and mixed subspaces.}}

%{\color{red}{\tt Comment(CEW) In the next two subsections I have made small edits to the wording butdidn't change any of the mathematics, except that I corrected the order of the arguments in $I^{-}(k,q,l)$ in the last displayed equation of this subsection.}}

\subsubsection{Imaginary subspace. Off diagonal entries.}
We now consider variations about the Gaussian subspace, beginning as before with the off-diagonal terms.
 For the purely imaginary case $(if_k,if_l)$, we have

\begin{equation}
\left . \frac{\partial^2 H(f(s))}{\partial  s_1 \partial s_2} \right |_{s_1=s_2=0} =
\end{equation}
\[
\int \int \frac{q^2}{4} |e^{it\Delta} f_0|^q \left ( i \frac{e^{it\Delta}f_{k}}{ e^{it\Delta}f_0} + c.c.\right )
\left (i \frac{ e^{it\Delta}  f_{l}}{ e^{it\Delta} f_0}+c.c. \right )-
\frac q 2  |e^{it\Delta} f_0|^q  \left ( i^2 \frac{ e^{it\Delta}   f_{k} e^{it\Delta}  f_{l}}{e^{it\Delta} f_0  e^{it\Delta} f_0 } + c.c. \right ) dx dt = 
\]
\[
= \frac{q^2}{4}(- I^+ - \bar I^+ + I^- + \bar I^-) - \frac q 2 (- I^+ - \bar I^+) = \frac{q^2}{2} I^-(k,l,q),
\]
which is the same expression as for the diagonal terms in the real subspace. Note that off diagonal terms involving the zero mode do not appear due to the fact that we consider only variations that preserve norm.

\subsubsection{Imaginary subspace. Diagonal terms.}
A similar calculation as above with $g_0$ deformed in the imaginary direction 
\[
g(s) = \sqrt{1-s^2} g_0 + is g_k, g(0) = g_0, g'(0)= ig_k, g''(0)=-g_0
\]

\be
\left . \frac{d^2}{d s^2}\right |_{s=0} |g|^q = \frac{q^2}{2} |g_0|^q \left (i \frac{g_k}{g_0} -i \frac{\bar g_k}{\bar g_0} \right )^2
+ \frac q 2  |g_0|^q \left (-1 + \frac{g_k^2}{g_0^2} +c.c. \right ).
\ee
The expression is the same as in the real case when $k,l\neq 0$
\be
{\mathfrak{H}}^{\mathbb I}_{kl} = \frac{q^2}{2} I^-(k,l,q)- \delta_{kl}I^-(0,0,q), k\geq 1, l\geq 1
\ee
and ${\mathfrak{H}}^{\mathbb I}_{kl} = 0$ if $k=0$ or $l=0$. \\

%{\color{red}{\tt Comment(CEW) Made the following paragraph its own subsection to agree with having separate subsections for the purely real and purely imaginary variations.}}

\subsection{Variations that mix real and imaginary directions}

Now, consider the mixed case, e.g. $(if_k, f_l)$ and we have 
\begin{equation}
\left . \frac{\partial^2 H(f(s))}{\partial  s_1 \partial s_2} \right |_{s_1=s_2=0}= 
\end{equation}
\[
\int \int \frac{q^2}{4} |e^{it\Delta} f_0|^q \left ( i \frac{e^{it\Delta}f_{k}}{ e^{it\Delta}f_0} + c.c.\right )
\left (\frac{ e^{it\Delta}  f_{l}}{ e^{it\Delta} f_0}+c.c. \right )-
\frac q 2  |e^{it\Delta} f_0|^q  \left ( i \frac{ e^{it\Delta}   f_{k} e^{it\Delta}  f_{l}}{e^{it\Delta} f_0  e^{it\Delta} f_0 } + c.c. \right ) dx dt =  
\]
\[
= \frac{q^2}{4} (i I^+ - i \bar I^-  + i I^{-} -i \bar I^+) -\frac q 2 (i I^+ -i \bar I^+) =0,
\]
since $I^-$ is real.

%{\color{red}{\tt Comment(CEW):  My changes to this section were a little more substantive, so Ihave put all sentences or equations with changes in them in blue font.}}

\subsubsection{Structure of the Hessian}
Recall that our goal is to show that the Gaussian critical point is at least a local minimizer.
To this end, we examine various approaches to showing that all the eigenvalues of the Hessian
matrix are negative.   As we have seen in the previous subsections, the structure
of the Hessian is the same in subspaces corresponding to variations in the purely real or purely
imaginary directions, (and the Hessian is zero in directions corresponding to mixed real/imaginary 
variations,)  so we focus just on variations in the purely real subspace.

First note that up to a constant multiplier, the matrix of partial derivatives $I^-(k,l,q)$ can be represented as Gram matrix of linearly independent  functions.

 Indeed, let 
\be
f_k(x,\tau) = c_k H_k(x) e^{i2\pi |k| \tau}=c_{k_1}c_{k_2}\dots c_{k_d} H_{k_1}(x_1) H_{k_2}(x_2) \dots H_{k_d}(x_d) e^{i2\pi (k_1+k_2 + \dots k_d)\tau}
\ee
be defined on $L^2({\mathbb R}^d\times [0,2\pi])$ with the inner product
\be
(f_k,f_l) = \frac{1}{2\pi}\int_{\mathbb R^d} \int_0^{2\pi}  c_k c_l H_k(x) H_l(x) e^{i 2\pi(|k|-|l|)\tau } e^{-q|x|^2/2}dx d\tau.
\ee
Thus, $I^-(k,l,q)$ is proportional to the matrix of inner products of linearly independent functions. 
By the property 
of Gramian matrices, the matrix is positive semi-definite.  Then, the Hessian is the difference of a positive semi-definite matrix and of a matrix proportional to the identity matrix. Therefore, one can conclude that the Hessian is nonpositive if the largest eigenvalue of the  $ I^-(k,l,q)$ is smaller than $(2/q) I^-(0,0,q)$.

For symmetric matrices, the largest eigenvalue is  bounded by the sum of the matrix elements over each column. Then we arrive at the following inequalities which would imply 
nonpositivity of the Hessian.

\begin{equation}
\sum_{|l|=|k|, l\neq 0} I^-(k,l,q)  \leq \frac{2}{q}  I^-(0,0,q)
\end{equation}
or equivalently 
\begin{equation}
\label{eq:comb_ineq}
\sum_{|l|=|k|, l\neq 0}  \prod_{j=1}^d  c_{k_j} c_{l_j}\int  e^{-q x^2/2}
  H_{k_j}(x)  H_{l_j} (x) dx \leq \frac 2 q  c_0^{2d} \left ( \int  e^{-q x^2/2} dx\right )^d,
\end{equation}
where multi-index $k = (k_1, k_2,  ..., k_d)$ is fixed and $k\neq 0$. \\

\subsubsection{Special cases} 
The inequality can be checked for some specific cases, e.g. $d=1$ which implies $q=6$. Then, we have
\be
c_n^2 \int e^{-3x^2} H^2_n(x) dx \leq \frac 2 6 c_0^2  \int e^{-3x^2} dx
\ee
or equivalently
\be
\frac{1}{2^n n!}\int e^{-3x^2} H^2_n(x) dx \leq \frac 1 3  \int e^{-3x^2} dx = \frac{\sqrt{\pi}}{3\sqrt{3}}.
\ee
Rearranging  and rescaling terms, we continue
\be
  \int H_n^2(x)e^{-3x^2} dx \leq  \frac{2^{n} n! \sqrt{\pi}}{3\sqrt{3}},
\ee
\be
 \frac{1}{\sqrt{3}}  \int H_n^2(x/ \sqrt{3})e^{-x^2} dx \leq \frac{2^{n} n! \sqrt{\pi}}{3\sqrt{3}}.
\ee
Using the product formula,
\be
H_n(\gamma x) = \sum_{i=0}^{\lfloor{\frac n 2}\rfloor} \gamma^{n-2i} (\gamma^2-1)^i {n\choose 2i}\frac{(2i)!}{i!} H_{n-2i}(x)
\ee
with $\gamma=1/\sqrt{3}$ we get rid of the integrals in the above inequality.

Square the product formula first
\be
H_n^2( x/\sqrt{3}) = \sum_{i=0}^{\lfloor{\frac n 2}\rfloor} (1/3)^{n-2i} (-2/3)^{2i} {n\choose 2i}^2 \left  (\frac{(2i)!}{i!} \right )^2 H_{n-2i}^2(x) + 
{\rm mixed} \,\, {\rm terms}.
\ee

Multiplying with $e^{-x^2}$ and integrating, so that all mixed terms drop out due to orthonormality, we obtain the inequality:

\be
 \frac{1}{\sqrt{3}}  \int H_n^2(x/ \sqrt{3})e^{-x^2} dx  =  \frac{1}{\sqrt{3}} \sum_{i=0}^{\lfloor{\frac n 2}\rfloor} (1/3)^{n-2i} (-2/3)^{2i} {n\choose 2i}^2 \left  (\frac{(2i)!}{i!} \right )^2 \int H_{n-2i}^2(x) e^{-x^2} dx = \nonumber
\ee

\be
 = \frac{1}{\sqrt{3}} \sum_{i=0}^{\lfloor{\frac n 2}\rfloor} (1/3)^{n-2i} (-2/3)^{2i} {n\choose 2i}^2 \left  (\frac{(2i)!}{i!} \right )^2
  \sqrt{\pi} 2^{n-2i} (n-2i)! \leq  \frac{2^{n} n! \sqrt{\pi}}{  3\sqrt{3}}
\ee
that should hold for all $n\geq 1$.

Taking advantage of a number of cancellations on the left hand-side of this inequality, 
we are reduced to proving  the inequality 
\be
\frac{1}{\sqrt{3}} \sum_{i=0}^{\lfloor{\frac n 2}\rfloor} (1/3)^{n}   \frac{n!}{(n-2i)! i!^2}   \leq  \frac{1}{3\sqrt{3}}
\ee
or equivalently
\begin{equation}\label{eq:Hessest}
\sum_{i=0}^{\lfloor{\frac n 2}\rfloor}   \frac{n!}{(n-2i)! \, i!^2}   \le  3^{n-1}.
\end{equation}

\begin{proposition}
The above inequality holds for any $n\geq 1$.
\end{proposition}
\begin{proof}

 Note that one can  easily check ``by hand'' 
that the two sides of \eqref{eq:Hessest} are equal for
$n=1,2$.   These correspond to the two zero eigenvalues of the Hessian evaluated at the
Gaussian discussed in Subsection \ref{s:Gaussian}. Thus, we can restrict consideration to $n \ge 3$.
First, recall multinomial formula
\be
3^n = (1+1+1)^n = \sum_{k_1+k_2+k_3=n}\frac{n!}{k_1! k_2! k_3!}
\ee
that will be used to prove the inequality. Assume first that $n$ is not a multiple of 3. Then 
\[
3 \sum_{i=0}^{\lfloor{\frac n 2}\rfloor}   \frac{n!}{(n-2i)! \, i!^2}  
\]
is just a part of the triple sum in the trinomial formula (since $n\neq 2i$ for any $i$). 

On the other hand, if $n$ is divisible by 3, and we apply the same argument then all terms can be 
matched with the corresponding ones in the trinomial formula except for 
\[
\frac{n!}{m!m!m!}
\]
that is multiplied by 3
 in the last sum but appears only once in the trinomial formula. 

Therefore, to prove the inequality, we need to bound two of these terms with some other terms in the trinomial formula, which are not matched yet with anything else.

Such terms are readily provided by 
\[
\frac{n!}{(m+1)! (m-1)! m!}
\] 
and there are 6 of them as all 3 components can be permuted. Thus, it suffices to verify
\be
2\frac{n!}{m!m!m!} \leq 6\frac{n!}{(m+1)! (m-1)! m!},
\ee
which is equivalent to $(m+1)\leq 3(m-1)$ implying the result if $m\geq 2$ or equivalently  for $n\geq 6$. This leaves only one case to consider $n=3$ which can be verified by direct
calculation.

\end{proof}

\begin{rmk}
One can derive similar combinatorial expressions in higher dimensions. They inequalities  appear to hold, too, but they are naturally  more difficult to prove.
\end{rmk}

\begin{figure}[ht]
\includegraphics[clip, trim= 2cm 6cm 2cm 6cm, height=100mm]{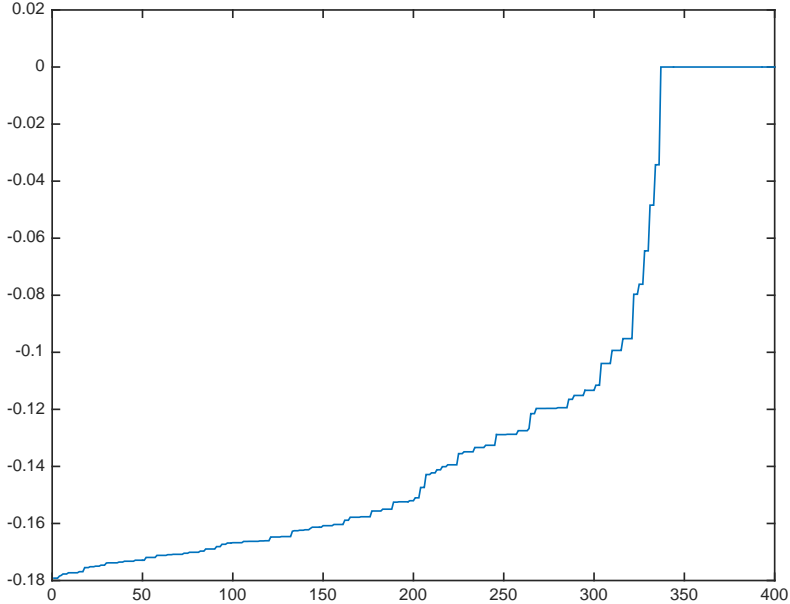}
\caption{Spectral gap in 3d case. This data was obtained  by evaluating \eqref{eq:comb_ineq} with $d=3$. There are several zero eigenvalues and all the remaining eigenvalues 
are negative separated by the gap about  0.03.}
\end{figure}

\section{Numerical Calculation of Hessian}
In this section we describe the details of our algorithm  that was used to compute the spectrum of the Hessian.

Note that in dimension 1, the Hessian becomes diagonal with only positive terms. The case of dimension 2 is already nontrivial numerically, but we already know from the previous work that Hessian is nonpositive.   In higher dimensions, the Hessian is a sparse matrix with some nonzero terms off diagonal. 

First introduce  normalization constants
\be
c_n^2 = \frac{1}{\sqrt{\pi}2^n n!} .
\ee
Now, introduce and compute the following integrals used to find components of the Hessian
\be
G(m,n) = c_m c_n \int e^{-q x^2/2} H_m(x)H_n(x)dx.
\ee
 We fix a large integer $N$, and compute $G(m,n)$ for all modes with $m$, $n$, less than or equal to
$N$.  We then use the Gramian structure of the Hessian to compute 
 the matrix of partial derivatives (Hessian) with $d$ being fixed and  
$q= 2+ 4/d$ 
\be
M(i,j) = G(k_1,l_1)\dots G(k_d,l_d)\cdot \chi_{|k|-|l|}.
\ee
We use the indicator function to avoid computing zero components where  $\chi_m=0$ if $m=0$ and $\chi_m=1$ if $m \neq 0$. We need to parametrize the values of $i,j$ to obtain a matrix and we do this  using a base $d$ expansion 
\be
i = k_1(N+1)^{d-1} + k_2(N+1)^{d-2} + \dots + k_d
\ee
\[
j = l_1(N+1)^{d-1} + l_2(N+1)^{d-2} + \dots + l_d.
\]

Next, using the calculations from the previous section about the structure
of the Hessian,  we subtract a diagonal matrix which is the identity matrix times the constant
\be
c = \frac{2}{q} (G(0,0))^d,  
\ee
so that the final expression for the Hessian components  is given by
\be
M(i,j)-\frac{2}{q} (G(0,0))^d  \delta_{ij}.
\ee

\section*{Acknowledgements} VZ thanks Simons foundation for partial support (\#278840 to Vadim Zharnitsky). The work of CEW was supported in part by the NSF through grant DMS-1311553.

%%%%%%%%%%%%%%%%%%%%%%%%%%%%%%%%%%%%%%%%%%%%%%%%%%
%%%%%%%%%%%%%%%%%%%%%%%%%%%%%%%%%%%%%%%%%%%%%%%%%%


\begin{thebibliography}{XXXX99}

\bibitem{AK2017} Albert, John, and Estapraq Kahlil. ``On the well-posedness of the Cauchy problem for some nonlocal nonlinear Schr\"odinger equations'', {\em Nonlinearity} 30.6 (2017): 2308.

\bibitem{bonforte} {Bonforte, M., Dolbeault, J., Gillo, G. and V\'azquez, J.L.} {``Sharp rates of decay of solutions to the nonlinear fast diffusion equation via functional inequalities"} {\emph{PNAS}}\textbf{107(38)}(2010), 16459--16464.


\bibitem{carlen} {Carlen, E., Carrillo, J. and Loss, M.} {``Hardy-Littlewood-Sobolev inequalities via fast diffusion flows"}{\emph{PNAS}} \textbf{107(46)}(2010), 19696-19701.

\bibitem{Carneiro} Carneiro, Emanuel, A sharp inequality for the Strichartz norm, {\em  International Mathematics Research Notices} 2009.16 (2009): 3127-3145.



\bibitem{ChristShao} Christ, Michael, and Shuanglin Shao, Existence of extremals for a Fourier restriction inequality,  {\em Analysis \& PDE} 5.2 (2012): 261-312.
	

\bibitem{ChristQuilodran} Christ, Michael and Quilodran Rene,  {Gaussians rarely extremize adjoint Fourier restriction inequalities for paraboloids,} {\em Proceedings of the American Mathematical Society } 142.3 (2014): 887-896.


\bibitem{faou} Faou, Erwan, Pierre Germain, and Zaher Hani. "The weakly nonlinear large-box limit of the 2D cubic nonlinear Schr\"odinger equation." {\em Journal of the American Mathematical Society} 29.4 (2016): 915-982.

\bibitem{foschi} Foschi, Diamiano,  Maximizers for the Strichartz inequality, {\em J. Eur. Math. Soc.} 8, 739--774 

\bibitem{FoschiSilva} Foschi, Damiano, and D. Oliveira e Silva, Some recent progress on sharp Fourier restriction theory, {\em Analysis Mathematica} 43.2 (2017): 241-265.

\bibitem{hdvz} Hundertmark, Dirk, and Zharnitsky, Vadim. ``On sharp Strichartz inequalities in low dimensions." {\em International Mathematics Research Notices} 2006, 
Art. ID 34080.

\bibitem{kunze} Kunze, Markus. ``On the existence of a maximizer for the Strichartz inequality." {\em Communications in Mathematical Physics} 243.1 (2003): 137-162.

\bibitem{strichartz}	{R. Strichartz}, {``Restrictions of Fourier transforms to quadratic surfaces and decay of solutions of wave equations."}{\emph{ Duke Mathematical Journal}} \textbf{44(3)} (1977), 705-714.





\end{thebibliography}
\end{document}